\documentclass[11pt, psfig]{article}
\usepackage{graphicx}
\usepackage{fancyheadings}
\usepackage{setspace}               
\include{definitions}               
\usepackage{amsfonts}
\usepackage{amsmath}
\usepackage{amssymb}
\usepackage{amsthm}
\usepackage{caption}
\usepackage{subcaption}
\usepackage{titlesec}
\usepackage{xcolor}
\usepackage[left=2cm,right=2cm,top=2cm,bottom=2cm]{geometry}
\usepackage{bm}
\DeclareGraphicsExtensions{.png,.pdf}

\newtheorem{theorem}{Theorem}[section]
\newtheorem{lemma}[theorem]{Lemma}
\newtheorem{proposition}[theorem]{Proposition}

\theoremstyle{definition}

\newtheorem{definition}[theorem]{Definition}
\theoremstyle{remark}
\newtheorem{remark}[theorem]{Remark}
\newtheorem{assumption}[theorem]{Assumption}

\numberwithin{equation}{section}

\def\sP{\mathbb{P}}

\def\sR{\mathbb{R}}
\def\sR{\mathbb{R}}

\def\sE{\mathbb{E}}

\newcommand{\Ac}{\mathcal{A}}

\newcommand{\be}{\begin{equation}}
\newcommand{\ee}{\end{equation}}
\newcommand{\bea}{\begin{eqnarray}}
\newcommand{\eea}{\end{eqnarray}}
\newcommand{\bean}{\begin{eqnarray*}}
\newcommand{\eean}{\end{eqnarray*}}
\def \proof{{\noindent \bf Proof}\quad}
\def \ep{\hbox{ }\hfill$\Box$}

\begin{document}

\title{Time consistent portfolio strategies for a general utility function}

\author{Oumar Mbodji\thanks{%
Independent Researcher}\\
\texttt{oumarsoule@gmail.com}
}

\maketitle

\noindent \textbf{Abstract.}

We study the Merton portfolio management problem within a complete market, non constant time discount rate and general utility framework. The non constant discount rate introduces time inconsistency which can be solved by introducing sub game perfect strategies. 
Under some asymptotic assumptions on the utility function, we show that the subgame perfect strategy is the same as the optimal strategy, provided the discount rate is replaced by the utility weighted discount rate $\rho(t,x)$ that depends on the time $t$ and wealth level $x$. A fixed point iteration is used to find $\rho$. The consumption to wealth ratio and the investment to wealth ratio are given in feedback form as functions of the value function.

\addcontentsline{toc}{section}{Abstract}

\vspace*{7mm}


\vspace*{5mm}


\section{Introduction}

This paper is a contribution to the analysis of time consistent stochastic optimization in a complete market. The stock price is modeled as a geometric Brownian motion with constant coefficients. The investor trades between the stock and a riskless security modeled with a constant interest rate. He aims to maximize his expected utility for consumption and terminal wealth. The utility function is discounted with a variable discount rate $\rho$ yielding time inconsistency.

This model is the simplest extension of the models in \cite{Ek_Pi1}, \cite{Ek_Mb_Pi1}, \cite{Mb_Pi1}
for a general utility function. We are still in a Markovian complete market setting and can use the theory of PDEs to build the solution. As in \cite{Mb_Pi1}, \cite{Bj_Kh_Mu1}, we define a value function that will solve an extended Hamilton-Jacobi-Bellman (HJB) equation. This paper will show that the extended HJB equation  can be solved by means of a fixed point problem. The work of \cite{Yo1} already solved the time consistency problem via a fixed point methodology. The novelty in our paper is that we define a marginal utility weighted discount rate that can be computed efficiently via Monte-Carlo simulation and that gives a better economic intuition of the solution.

Our present paper uses the methodology of \cite{Na_Za1} and works directly with the marginal value function $V_x(t,x)$ instead of the original value function $V(t,x)$.

Moreover, time-consistent policies are constructed in a feedback form from the first order conditions in the extended HJB equation.

In the rest of this section, we present the market model, define the discount function and the maximal expected utility problem. In Section 2, we state the main results and introduce the fixed point problem. Section 3 describes the algorithm for constructing the subgame perfect strategies.

\subsection{The Financial Market}
Consider a financial market consisting of a savings account and one stock
(the risky asset).  We assume a benchmark deterministic time horizon $T$.
The stock price per share follows an exponential Brownian motion 
\begin{equation}\label{eq121_1}
dS_t=S_t\left[ \mu \,dt+\sigma \,dW(t)\right] ,\quad 0\leq t\leq T
\end{equation}%
where $\{W(t)\}_{t\geq 0}$ is a $1-$dimensional Brownian motion on a
filtered probability space, $(\Omega ,\{\mathcal{F}_{t}\}_{0\leq t\leq T},\mathbb{P}).$  The filtration $\{\mathcal{F}_{t}\}$ is the completed
filtration generated by $\{W(t)\}$. 
The savings account accrues interest at the riskless rate $r>0$.

 Let us denote \be \label{eq121_2}
 \theta \triangleq \frac{\mu
-r}{\sigma}\ee
 \emph{the market price of risk}.

We place ourselves in a Markovian setting.

There is one agent who is continuously investing in the stock, is using 
the money market, and consuming.
At every time $t$, the agent chooses $\pi(t)$, the ratio of wealth invested in the
risky asset and  $c(t)$ the ratio of wealth consumed. Given an adapted process $\{\pi (t),c(t)\}_{0\leq t\leq T}$, the
equation describing the dynamics of wealth ${X^{\pi ,c}(t)}$ is given by : 
\begin{eqnarray}\label{eq121_3}
dX^{\pi ,c}(t) &=&[r-c(t)+\sigma\theta\pi(t)]X^{\pi ,c}(t)dt+\sigma\pi
(t)X^{\pi,c}(t)dW_t   \\
X^{\pi ,c}(0) &=&x_0 \nonumber
\end{eqnarray}%
the initial wealth $x_0$ being exogenously specified.

\subsection{Time preferences and risk preferences }
As seen in the introduction, the time preference reflects the economic agent's preference for immediate utility over delayed utility. We now define discount functions and discount rates.
\begin{definition}
A discount function $h:D=\left\{0\leq t\leq s\leq T\right\}\rightarrow \sR$ is a continuous, positive function satisfying $h(t,t)=1$.
$h(t,s)$ is the discount factor between time $t$ and $s$ with $t\leq s$. 
\end{definition}
\begin{remark}
We take a discount form to be of the general form $h(t,s)$ because, as noted in \cite{Pi_Zh1}, \cite{Mb_Pi1} and  \cite{Ek_Mb_Pi1}, we have to account for stochastic time horizons $T$ which could be the time of death of the agent.
In that case, the discount function has to be transformed and will take the general form $h(t,s)$.
We can normalize by dividing $h(t,s)$ by $h(t,t)$.
\end{remark}
\begin{assumption}
The discount function $h$ satisfies \begin{itemize}
\item \be
0< \inf_{(t,s)\in D} h(t,s) \leq \sup_{(t,s)\in D} h(t,s) <\infty \ee  \\
\item The functions  $\frac{\partial h(t,s)}{\partial s}, \frac{\partial h(t,s)}{\partial t}, \frac{\partial^2 h(t,s)}{\partial^2 t}$ are continuous and bounded.
\end{itemize}

The discount rate is defined as
\begin{equation}\label{eq121_5}
\rho_{h}(t,s) = \frac{\partial h(t,s)}{\partial t}\times \frac{1}{h(t,s)}
\end{equation}

\end{assumption}

In the case $h(t,s)=H(s-t)$ for a certain $C^1$ function $H$ on $[0,T]$, we get:
$\rho_h(t,s)=-\frac{H'(s-t)}{H(s-t)}$.

Next, we define the agent's \emph{risk preferences}.
An economic agent will have satisfaction $U(C)$ from consuming an amount $C$. 

\begin{assumption}\label{ass_U}
The utility function $U$ is twice continuously differentiable.

We assume:
\begin{enumerate}
\item $U$ is $C^2$ on $(0,\infty)$, strictly increasing, strictly concave and satisfies the Inada condition: 

$U'(0)=\infty$, $U'(\infty)=0$.
\item Let $$\mathcal{R}_1(x):=-\frac{x U''(x)}{U'(x)}$$ be its relative risk aversion coefficient. $\mathcal{R}_1(x)$ is positive and bounded away from 0. That is, there exists $r_1, r_2 $ positive numbers such that:
\be 
r_1 \leq \mathcal{R}_1(x) \leq r_2 \quad \forall x>0.
\ee
\end{enumerate}

\end{assumption}
\begin{remark}

 Note that condition 2 of Assumption \ref{ass_U} implies that the asymptotic elasticity defined in \cite{Kr_Sc1}, \cite{Sc1} \quad $AE(U)=\lim\inf_{x\rightarrow \infty} \frac{xU'(x)}{U(x)} <1$.

\end{remark}

Utility functions of the form $ \alpha \frac{x^{\gamma_1}}{\gamma_1}+ (1-\alpha) \frac{x^{\gamma_2}}{\gamma_2}$, $\gamma_1\neq 0, \gamma_2 \neq 0$ satisfy assumption \ref{ass_U}.

The exponential utility function does not satisfy the assumption.

\subsection{The intertemporal utility}
 Let us now define the admissible strategies.

\begin{definition}
 \label{def121_1} An $\mathbb{R}^{2}$-valued
stochastic process $\{\pi (t),c(t)\}_{0\leq t\leq T}$ is called an
admissible strategy process if:

\begin{enumerate}
\item it is progressively measurable with respect to the sigma algebra $%
\sigma (\{\ W(t)\}_{0\leq t\leq T})$,
\item it satisfies \begin{equation}\label{eq12_3}
c(t)\geq 0\,\, \text{ for all $t$ almost surely and } {X}^{\pi ,c}(T)\geq 0,\text{ almost surely }
\end{equation} 
\item moreover, we require that 
\begin{equation}  \label{eq121_4}
\mathbb{E} \sup_{0\leq s\leq T} |U(c(s)X^{\pi,c}_s)|<\infty \; \; , \; \;
 \mathbb{E}
\sup_{0\leq s\leq T} |U({X}^{\pi,c}_s)|<\infty, \, \text{ a.s.%
}
\end{equation}

\end{enumerate}
\end{definition}

The last set of inequalities is purely technical and is satisfied for e.g.\ bounded strategies.

In order to evaluate the performance of an investment-consumption
strategy the investor uses an expected utility criterion. For an
admissible strategy process $\{{\pi }(s),c(s)\}_{s\geq 0}$ and its
corresponding wealth process $\{X^{\pi ,c}(s)\}_{s\geq 0}$, we denote the
inter-temporal utility  by 
\begin{eqnarray}  \label{eq122_1}
J( t,x,\pi ,c)&=& \mathbb{E}_t\bigg[ \int_{t}^{T} h(t,s)
U(c(s)X^{\pi,c}(s))\,ds + h(t,T) {U}%
(X^{\pi,c}(T))\bigg]  \nonumber
\end{eqnarray}

A natural objective for the decision maker is to maximize the above expected
utility criterion.

As shown in \cite{Ek_Mb_Pi1}, the optimal strategy is time inconsistent when the discount function is not of the exponential form.

The decision-maker could implement two types of strategies. He could \emph{precommit} at time $t_0=0$ to follow the optimal strategy and stay with it until time $T$.
Or he could implement a time consistent strategy  that takes into account the fact that the decision maker's preferences
will change in the future. 

Time consistent problems have been studied extensively in the literature. 
\subsection{Subgame perfect strategies and value function}

We now introduce a special class of time consistent strategies, which can also be called \emph{subgame perfect strategies}. That is, we consider that the
decision-maker at time $t$ can commit his successors up to time $t+\epsilon 
$, with $\epsilon \rightarrow 0$, and we seek strategies which are
optimal to implement right now conditioned on them being implemented in the
future. This is made precise in the following formal definition.

\begin{definition}\label{def13_1}
An admissible trading strategy $\{\bar{\pi}(s),\bar{c}(s)\}_{0\leq s\leq T}$ is a \emph{subgame perfect strategy} if there exists a map $%
G=(G_{\pi}, G_{c}):[0,T]\times \mathbb{R_{+}}\rightarrow \mathbb{R}\times
\lbrack 0,\infty )$ such that for any $t\in [0,T] $ and wealth level $x >0$ at time $t$ 
\begin{equation}\label{eq13_1}
{\lim \inf_{\epsilon \downarrow 0}}\frac{J(t,x,\bar{\pi},\bar{c}%
)-J(t,x,\pi _{\epsilon },c_{\epsilon})}{\epsilon }\geq 0,
\end{equation}%
where:%
\begin{equation}\label{eq13_2}
\bar{\pi}(s)={G_{\pi}(s,\bar{X}(s))},\quad \bar{c}(s)=G_{c}(s,\bar{X}%
(s))
\end{equation}%
and the wealth process $\bar{X}(s):= X^{\bar{\pi},\bar{c}}(s) $ is a solution of
the stochastic differential equation (SDE): 
\begin{eqnarray}\label{eq13_3}
d\bar{X}(s)&=&\bar{X}(s)[r+\sigma\theta G_{\pi}(s,\bar{X}(s))-G_{c}(s,\bar{X}(s))]ds +\sigma G_{\pi}(s,\bar{X}(s))\bar{X}(s) dW(s)
\end{eqnarray}

The process $\{{\pi }_{\epsilon }(s),{c}_{\epsilon
}(s)\}_{s\in \lbrack t,T]}$ mentioned above is another
investment-consumption strategy defined by 
\begin{equation}\label{eq13_4}
\pi _{\epsilon }(s)=%
\begin{cases}
G_{\pi}(s, X_{\epsilon}(s)),\quad s\in \lbrack t,T]\backslash E_{\epsilon ,t} \\ 
\pi (s),\quad s\in E_{\epsilon ,t},%
\end{cases}
\end{equation}%
\begin{equation}\label{eq13_5}
c_{\epsilon }(s)=%
\begin{cases}
G_c(s,X_{\epsilon}(s)),\quad s\in \lbrack t,T]\backslash E_{\epsilon ,t} \\ 
c(s),\quad s\in E_{\epsilon ,t},%
\end{cases}
\end{equation}%
with $E_{\epsilon,t}=[t,t+\epsilon],$ and $\{{\pi}(s),{c}(s)\}_{s\in E_{\epsilon,t} }$ is any strategy for which $\{{\pi}%
_{\epsilon}(s),{c}_{\epsilon} (s)\}_{s\in[t,T]}$ is an
admissible policy.
$X_{\epsilon}$ is defined on $[t+\epsilon,T]$ by the SDE:
\begin{eqnarray}\label{eq13_6}
\!\!\!\!\!\!d X_{\epsilon}(s)&=&[r-c_{\epsilon}(s)+\sigma\theta\pi_{\epsilon}(s)]X_{\epsilon}(s)ds+\pi_{\epsilon}
(s)\sigma X_{\epsilon}(s)dW(s)  \notag   \\
X_{\epsilon}\left( t+\epsilon\right) \!\!\!\! \!&=&\!\!\!\!\! X^{\pi,c}(t+\epsilon).
\end{eqnarray}

\end{definition}

Dynamic programming is a very convenient way of formulating a large set of dynamic problems in financial economics. Most properties of this tool are well established and understood.
In dynamic programming, we introduce an object called the \emph{value function} that is obtained by evaluating a certain functional at our candidate solutions. The solutions of the dynamic programming problem are then the solutions of a certain equation called HJB and can be expressed entirely in terms of the value function and its derivatives.
In optimization problems, the value function is the optimal value an agent can derive from his maximization process.
 The paper \cite{Ek_Mb_Pi1} uses the value function methodology to characterize subgame perfect strategies. We will see that the value function can be written as a function $V(t,x)$ of time $t$ and wealth $x$ and this allows us to find subgame perfect strategies 
in a feedback form. For fixed $t, x$, the strategy $(\bar{\pi}, \bar{c})$ can be expressed as deterministic functions of $V$ and its derivatives with respect to $x$.
 We start with a definition.
\begin{definition}\label{def_extendedHJB}
 Let $V : [0,T] \times (0,\infty) \rightarrow \mathbb{R}, (t, x)\mapsto V(t,x)$ be a $C^{1,2}$ function. Suppose $V$ is strictly increasing, concave in the $x$ variable.
Let $I(x):=(U')^{-1}(x)$ the inverse marginal utility.
Suppose that $\{\bar{\pi }(s),\bar{c}(s), \bar{X}(s)\}_{s\in [0,T]}$ are subgame perfect strategies with the corresponding map

\begin{equation}\label{eqbarpi_barc}
\bar{\pi}(s)={G_{\pi}(s,\bar{X}(s))}\;\;,\quad \bar{c}(s)=G_{c}(s,\bar{X}%
(s)), 
\end{equation}%
where 
\begin{equation}\label{eqGpi_Gc}
\!\!\!\!\! G_{\pi}(t,x)=-\frac{\theta \frac{\partial V}{\partial x}(t,x)}{\sigma x
\frac{\partial ^{2}V}{\partial x^{2}}(t,x)}\;\;,\ \ G_{c}(t,x)=\frac{1}{x}\times I\left( \frac{%
\partial V}{\partial x}(t,x)\right) \end{equation}%
and $\bar{X}(s)$ is the wealth process given by:
\begin{eqnarray}\label{eqdbarX}
d\bar{X}(s)&=&[r+\sigma\theta G_{\pi}(s,\bar{X}(s))-G_{c}(s,\bar{X}(s))]\bar{X}(s)ds +\sigma G_{\pi}(s,\bar{X}(s))\bar{X}(s)dW(s). 
\end{eqnarray}

We shall say that $V$ is a value function if for all $(t, x)\in [0,T]\times (0,\infty) $, we have:%
\begin{equation}\label{eqV_equals_J}
V(t,x)=J(t,x,\bar{\pi},\bar{c})  
\end{equation}

\end{definition}

The economic interpretation is very natural: if one applies the Markov
strategy associated with $V$ by the relations [\eqref{eqbarpi_barc}, \eqref{eqGpi_Gc}, \eqref{eqdbarX}, \eqref{eqV_equals_J}] and computes the
corresponding value of the investor's criterion starting from $ X_t=x$ at time $t$%
, one gets precisely $V\left( t,x\right) $. In other words, this is
fundamentally a fixed-point characterization. 

In the next section, we give the main results of this paper.
\section{Main Results}

\subsection{The extended HJB}
As in \cite{Ek_Pi1}, \cite{Ek_Mb_Pi1}, \cite{Mb_Pi1}, the time consistent policy can be determined by solving an HJB equation with a non local term.
\cite{Bj_Kh_Mu1} provides a more general setting and shows that the value function solves an extended HJB system along with a verification theorem. The authors \cite{Le_Pu1}, \cite{Le_Pu2} have shown the well posedness of the extended HJB system. Our result is new in that it deals with a general utility function and provides a way to construct the solution.

The following proposition gives the subgame perfect strategies in terms of the value function.
\begin{proposition}\label{policies_TC}
If the extended HJB \eqref{extended_HJB} has a $C^{1,2}$ solution $V$, then
the subgame perfect  strategies are given by:
\begin{eqnarray}
\bar{c}(t,x)&=&\frac{I(V_x(t,x))}{x}\\ \label{eq151_2}
\bar{\pi}(t,x)&=&-\frac{\theta V_x(t,x) }{\sigma x V_{xx}(t,x)} \label{eq151_3}
\end{eqnarray}

\end{proposition}

The proof comes from a simple calculation of the first order conditions for the quantity $$(\bar{\pi},\bar{c})= {\arg\max}_{(\pi,c) \;admissible} \; \{\mathcal{A}^{\pi,c}V+U(xc) \}$$. 

\begin{theorem}[Extended HJB]\label{thm151_1}
 Let $V:[0,T]\times  (0,\infty)\rightarrow \mathbb{R}$ be a $C^{1,2}$ function.
 
 Suppose $(\bar{\pi}, \bar{c})$ is an admissible Markovian policy and that 
 \begin{itemize}
\item $V$ solves the extended Hamilton Jacobi Bellman equation : 
 
 \begin{eqnarray}  \label{extended_HJB}
&&\frac{\partial V}{\partial t}(t,x)+\sup_{(\pi ,c)\; \textrm{admissible}}\big\{ \mathcal{A}^{\pi,c}V(t,x)+U (xc)\big\} \\
&=&\sE_t \left[\int_t^T \frac{\partial h(t,s)}{\partial t}U(\bar{c}(s)X^{\bar{\pi},\bar{c}}(s) ) ds + \frac{\partial h(t,T)}{\partial t}U(X^{\bar{\pi},\bar{c}}(T) )\right]\nonumber
\end{eqnarray}
where 
\begin{equation}\label{eq14_5}
 \mathcal{A}^{\pi,c}V(t,x)= (r+\sigma\theta \pi -c)x
\frac{\partial V}{\partial x}(t,x)+\frac{1}{2}\sigma ^{2}x^2\pi ^{2}%
\frac{\partial ^{2}V}{\partial x^{2}}(t,x)
\end{equation}
along with the boundary condition $V(T,x)= U(x).$ 
\item $(\bar{\pi},\bar{c})$ satisfies:
\be  \label{eq151_2b}
(\bar{\pi},\bar{c}) = {\arg\max} \{ \mathcal{A}^{\pi,c}V(t,x)+U (xc) ; (\pi ,c)\; \textrm{ admissible }  \}
\ee
\end{itemize}
Then $V$ is a value function. Moreover $\{\bar{\pi}(s), \bar{c}(s), \bar{X}(s)\}$
 is a subgame perfect strategy (cf.\ Definition \ref{def13_1}).

\end{theorem}
Theorem \ref{thm151_1} is proven in \cite{Bj_Kh_Mu1}. It is called a verification theorem because it allows us to check if a given value function is actually a subgame perfect strategy.

Next, we define a strategy dependent discount rate that we call utility weighted discount rate.

Substituting $\bar{c}, \bar{\pi}$ by the expressions in proposition \ref{policies_TC}, the extended HJB becomes:

\begin{eqnarray*}\label{Ext_HJB2}
&&V_t+\big(rx-\frac{\theta^2 V_x}{2 V_{xx}}-I(V_x)\big)V_x+U(I(V_x))\\
&=&\sE_t \bigg[ \int_t^T \frac{\partial h(t,s)}{\partial t}U(I(V_x(s,\bar{X}_s))ds+\frac{\partial h(t,T)}{\partial t}U(\bar{X}_T)\bigg]
\end{eqnarray*}

In the spirit of \cite{Na_Za1}, 
the $x$ derivative of the extended HJB gives:
\begin{eqnarray} \label{equation_Vx}
&&V_{tx}+(r-\theta^2)V_x+(rx-I(V_x))V_{xx}+\frac{(\theta V_x)^2V_{xxx}}{2V_{xx}^2}(t,x) \nonumber \\
&& =\frac{\partial}{\partial x}\sE_t \bigg[ \int_t^T \frac{\partial h(t,s)}{\partial t}U(\bar{c}_s\bar{X}_s) ds+\frac{\partial h(t,T)}{\partial t} U(\bar{X}_T) \bigg]\\
&& = \rho(t,x) V_x(t,x)
\end{eqnarray}
and equation \eqref{equation_Vx} is known as the marginal extended HJB.

The quantity $\rho(t,x)$ is defined as:

\begin{equation}\label{quotient_rho}
\rho(t,x)=\frac{  \int_t^T \frac{\partial h(t,s)}{\partial t} \frac{\partial }{\partial x}\sE_t[U(\bar{c}_s\bar{X}_s)] ds+\frac{\partial h(t,T)}{\partial t} \frac{\partial }{\partial x}\sE_t[U(\bar{X}_T))] }{  \int_t^T  h(t,s) \frac{\partial }{\partial x}\sE_t[U(\bar{c}_s\bar{X}_s)] ds+ h(t,T)\frac{\partial }{\partial x}\sE_t[U(\bar{X}_T))]}
\end{equation}
Note that the denominator in the expression above is equal to $V_x(t,x)>0$.
The quantity $\rho$ is an average discount rate weighted by the marginal utility.
In particular, $\rho(t,x) \in \left[\inf  \frac{\frac{\partial h}{\partial t}}{h}, \sup\frac{\frac{\partial h}{\partial t}}{h}\right]$.

In what follows, we write
\be
v(t,x) := V_x(t,x)
\ee
Since all the terms in \eqref{equation_Vx} involve $v$ and its derivatives, we get a PDE for $v$:

\begin{eqnarray}\label{PDE_v}
&&v_{t}(t,x)+(r-\theta^2)v+(rx-I(v))v_{x}+\frac{\theta^2}{2}\frac{v^2v_{xx}}{v_{x}^2}  =\rho(t,x) v(t,x)\\
&& v(T,x) = U'(x)
\end{eqnarray}
$v$ satisfies a non linear parabolic PDE that is possibly degenerate.
A way around it, is to work with the $x$ inverse of $v(t,x)$ and get a linear PDE.

Assuming that $V$ is strictly concave and $C^2$ as a function of $x$ yields a unique $p(t,x)$ such that

\begin{equation}\label{Vx_p}
v(t,p(t, x))=x
\end{equation}

$p$ is the $x$-inverse of the marginal value function $v=V_x$.
We calculate partial derivatives from equation \eqref{Vx_p}: 
 \begin{equation}\label{relations_vp}
v_{x}=\frac{1}{p_x} \; , \quad v_{xx}=-\frac{p_{xx}}{p_x^3} \; , \quad v_{t}=-\frac{p_t}{p_x}
\end{equation}
where $v$ is evaluated at $(t, p(t,x))$ and $p$ is evaluated at $(t, x)$.

In terms of $p$, the marginal HJB \eqref{equation_Vx} becomes
\begin{eqnarray*}
\rho(t,p(t,x)) x&=&-\frac{p_t(t,x)}{p_x(t,x)}+(rp-I(x))\frac{1}{p_x}+(r-\theta^2)x-\frac{p_{xx}}{2p_x^3}(-\theta xp_x)^2
\end{eqnarray*}

\begin{eqnarray}\label{PDE_p}
p_t+\frac{\theta^2}{2}x^2p_{xx}
-rp+I(x)+(\rho(t,p(t,x))+\theta^2-r)xp_x (t,x)&=&0 \\
p(T,x) &=& I(x)
\end{eqnarray}
$p$ solves a linear parabolic PDE which can be solved using standard techniques.

We can change the variable to $ y=\log x$  
\begin{definition}
For fixed $(t,y)\in [0,T]\times \sR$, define

 $I_0(y)=I(e^y)$ and $U_0(y):=U(I(e^y))$,  
  $\bar{p}(t,y)=p(t,e^y)$ and $\bar{\rho}(t, y)=\rho(t,p(t,e^y))$.
  The functions 
$\bar{\rho}$ and $\rho$ satisfy the relation
\be
\rho(t,x) = \bar{\rho}(t, \log v(t,x)) 
\ee
 \end{definition}
 Since $\bar{p}_{y}(t,y)=e^y p_x(t,e^y) =xp_x(t,x)$ and $\bar{p}_{yy}(t,y)=e^{2y} p_{xx}(t,e^y)+e^yp_x(t,e^y) =x^2p_{xx}(t,x)+xp_x(t,x)$. 
  We get
\begin{eqnarray}
\bar{p}_{t}+\frac{\theta^2}{2}\bar{p}_{yy}+\big(\frac{\theta^2}{2}+\bar{\rho}(t,y)-r\big)\bar{p}_y -r\bar{p} +I_0(y)&=&0 \label{PDE_barp}\\
\bar{p}(T,y)&=&I_0(y)
\end{eqnarray}

We make the following assumptions about the function $\bar{\rho}$.

\begin{assumption}
The functions $\bar{\rho}(t,y)$ and $\bar{\rho}_y(t,y)$ are continuous and bounded.
\be
\forall t\in[0,T], y\in \sR: |\bar{\rho}(t,y)|\leq ||\rho|| \quad \textrm{ and } \quad |\bar{\rho}_{y}(t,y)| \leq  \kappa
\ee
\end{assumption}
We will construct such a function at the end of this paper.

We can still solve the equation above using the Feynman-Kac formula if we assume $\bar{\rho}$ is known 
and is $C^{1}$ in the variables $t, y$.                                                                                                                                                             

We define the process $\bar{Y}_s$ by the SDE:

\begin{equation}\label{SDE_barY}
d\bar{Y}_{s}=(-\frac{\theta^2}{2}-r+\bar{\rho}(s,\bar{Y}_{s}))ds-\theta dW_{s} \;\; , \;\;  \bar{Y}_t=y 
\end{equation}
By Feynman Kac's formula:

\begin{eqnarray*}
\bar{p}(t,y) &=& \sE\left[  e^{-r(T-t) } I_{0}(\bar{Y}_T+\theta^2(T-t))+\int_{t}^T e^{- r(s-t) } I_{0}(\bar{Y}_s+\theta^2(s-t))ds  \;\; \bigg|\; \bar{Y}_t=y \right]
\end{eqnarray*}

We want to get bounds for $p$, $\bar{p}$ and $v$. 
In the spirit of \cite{Na_Za1}, we start by introducing some function spaces.
\begin{definition}
For each $0<\nu_1\leq \nu_2$, denote by $\mathcal{D}_0(\nu_1, \nu_2)$ the space of functions $F:(0,\infty)\rightarrow (0,\infty)$ of class $C^1$ such that
\be \label{ineq_nu}
\nu_1 \leq -\frac{xF'(x)}{F(x)}\leq \nu_{2}\quad \forall x\in (0,\infty)
\ee
\end{definition}
We will need the following lemma:
\begin{lemma}\label{lemma_D0}
Suppose $F\in \mathcal{D}_0(\nu_1, \nu_2)$ and let $f:(0,\infty)\rightarrow (0,\infty)$ be its inverse i.e. $F(f(x))=x$.
 
 Then $f\in \mathcal{D}_0\left(\frac{1}{\nu_2}, \frac{1}{\nu_1}\right)$.
\end{lemma}
\proof
Replace $x$ by $f(x)$ in \eqref{ineq_nu} to get:
$\nu_1 \leq -\frac{f(x)F'(f(x))}{F(f(x))}\leq \nu_{2}$. Using the relations $F(f(x))=x$ and $f'(x)=\frac{1}{F'(f(x))}$, we get:
$\nu_1 \leq -\frac{f(x)}{xf'(x)}\leq \nu_{2}$ i.e.  $\frac{1}{\nu_2} \leq -\frac{x f'(x)}{f(x)}\leq \frac{1}{\nu_1}$.
\ep

\begin{proposition}\label{p_v_D0}
For $t\in [0,T]$, let \be
r_1(t) := r_1 e^{-\kappa (T-t)} \text{ and } r_2(t) := r_2 e^{\kappa (T-t)}.
\ee
Then $v(t, .) \in \mathcal{D}_0(r_1 (t), r_2(t))$ and
$p(t, .) \in \mathcal{D}_0\left( \frac{1}{r_2(t)},     \frac{1}{r_1(t)}\right)$.

\end{proposition}
The proof is in the appendix.

\begin{proposition}\label{bounds_barp}
The function $\frac{1}{I_0(y)}\times\frac{\partial \bar{p}(t,y)}{\partial t}$ is bounded and for $k=0, 1, 2$, the functions $\frac{1}{I_0(y)} \times \frac{\partial^k \bar{p}(t,y)}{\partial y^k}$ are bounded functions of $(t,y)\in [0,T]\times \sR$. Furthermore there are positive continuous functions $r_3(t), r_4(t)$ such that 
\be \label{ineq_barp_over_J}
r_3(t) \leq \frac{\bar{p}(t,y)}{I_{0}(y)} \leq r_4(t)
\ee
\end{proposition}

The proof will be in the appendix.

The function $y\mapsto \bar{p}(t,y)$ is strictly decreasing. Since $\lim_{y\rightarrow \infty} I_0(y) =0$ and $\lim_{y\rightarrow -\infty} I_0(y) =\infty$, we get the following result:
\begin{proposition}
For every $t\in [0,T]$, the function $y\mapsto \bar{p}(t,y)$ defines a bijection from $\sR$ to $(0,\infty)$. And we have: 
$\lim_{y\rightarrow \infty} \bar{p}(t,y) =0$ and $\lim_{y\rightarrow -\infty} \bar{p}(t,y) =\infty$.
\end{proposition}

\begin{proposition}\label{bounds_candpi}
The quantities \begin{eqnarray*}
 \bar{\pi}(t,x)&=&-\frac{\theta v(t,x)}{\sigma x v_{x}(t,x)} \\
\bar{c}(t,x)&=&\frac{I(v(t,x))}{x}
\end{eqnarray*}
are bounded by 
 \begin{equation}\label{bounds_pi_c}
 \frac{\theta r_1(t)}{\sigma} \leq \bar{\pi}(t,x)\leq \frac{\theta}{\sigma r_1(t)}\quad  \textrm{  and }\quad \frac{1}{r_4(t)}\leq \bar{c}(t,x) \leq \frac{1}{r_3(t)}
 \end{equation}
\end{proposition}
\proof
Since $v(t, .) \in \mathcal{D}_0(r_1(t), r_2(t))$,  $$r_1(t) \leq -\frac{xv_x(t,x)}{v(t,x)} \leq r_2(t).$$ So $$\frac{\theta}{\sigma} r_1(t) \leq \bar{\pi}(t,x)\leq \frac{\theta}{\sigma}\frac{1}{r_1(t)} $$

By making the change of variables $x=\bar{p}(t,y)$, we get:

$$\bar{c}(t,x)=\frac{I(v(t,x))}{x}= \frac{I(v(t,\bar{p}(t,y)))}{\bar{p}(t,y)}=\frac{I_0(y)}{\bar{p}(t,y)}$$

Since by Proposition \ref{bounds_barp}, $$r_3(t)\leq \frac{\bar{p}(t,y)}{I_{0}(y)}\leq r_4(t)$$
 we conclude that
 $$
\frac{1}{r_4(t)}\leq \bar{c}(t,x) \leq \frac{1}{r_3(t)}
$$

\ep

Next, we give a characterization of the wealth process.

\begin{theorem}\label{barXbarp}
Let $t\in [0,T]$ and $x>0$.
The wealth process $\bar{X}_s := X^{\bar{\pi},\bar{c}}(s)$ is given by
$$\bar{X}(s)=X^{\bar{\pi},\bar{c}}(s)=\bar{p}(s,\bar{Y}(s))$$ where
$\bar{X}(t)=x=\bar{p}(t,y)$,  $y=\log v(t,x)$ and $\bar{Y}$ satisfies the SDE \eqref{SDE_barY}.
 \end{theorem}
 
The proof is in the appendix. 

\subsection{Calculation of the quantity $\bar{\rho}$}

We use Theorem \ref{barXbarp} to get $$\bar{c}_s\bar{X}_s = I(v(s,\bar{X}(s))=I(e^{\bar{Y}_s})=I_0(\bar{Y}_s).$$
$$\bar{X}_T=\bar{p}(T, \bar{Y}_T)= I_0(\bar{Y}_T)$$
 and the following: $$\frac{\partial }{\partial y}= \frac{\partial x}{\partial y}. \frac{\partial }{\partial x} = \bar{p}_y(t,y) . \frac{\partial }{\partial x}$$

 \begin{eqnarray}\label{equation_rho}
\rho(t,\bar{p}(t,y)) &=& \frac{\frac{\partial }{\partial y} \sE_{t}\bigg[\int_{t}^{T}\frac{\partial h(t,s)}{\partial t}U_0(\bar{Y}_s)ds+\frac{\partial h(t,T))}{\partial t}U_0(\bar{Y}_T)\bigg]}{\frac{\partial }{\partial y} \sE_{t}\bigg[\int_{t}^{T}h(t,s)U_0(\bar{Y}_s)ds+h(t,T)U_0(\bar{Y}_T)\bigg]}=\bar{\rho}(t,y)
\end{eqnarray}

Thus, $\bar{\rho}$ is solution of a fixed point operator $F$. We first start by a definition.

\begin{definition}\label{space_functions}
We define the space $\mathbb{B}$ the space of functions $\phi: [0,T]\times \sR \mapsto \sR$ such that :
\begin{itemize}
\item $|\phi(t,y)| \leq ||\rho||$ $\;\; , \; \forall t,y$.
 \item $(t,y)\mapsto \phi(t,y)$ , $(t,y)\mapsto \phi_y(t,y)$ are bounded continuous functions.
 \end{itemize}
 
Let $\kappa>0$, we define $$\mathbb{B}_{\kappa} = \{ \phi\in \mathbb{B}\;\; | \;\; \forall t\in [0,T], y\in \sR: |\phi_y(t,y)| \leq \kappa  \}$$

\end{definition}

For $\phi\in \mathbb{B}$, 
$Y^{\phi}$ is the solution of the SDE:
 \begin{equation}
dY_{s}^{\phi}=(-\frac{\theta^2}{2}-r+\phi(s,Y_{s}^{\phi}))ds-\theta dW_{s} \;\; , t\leq s\leq T \;\; ,  Y_t^{\phi}=y 
\end{equation}
For $t\in [0,T], y\in \sR$, define  
 \begin{equation}
F[\phi](t,y)= \frac{\frac{\partial }{\partial y} \sE_{t}\bigg[\int_{t}^{T}\frac{\partial h(t,s)}{\partial t}U_0(Y_s^{\phi})ds+\frac{\partial h(t,T))}{\partial t}U_0(Y_T^{\phi})\bigg]}{\frac{\partial }{\partial y} \sE_{t}\bigg[\int_{t}^{T}h(t,s)U_0(Y_s^{\phi})ds+h(t,T)U_0(Y_T^{\phi})\bigg]}
 \end{equation}
\begin{theorem}
$\bar{\rho}$ is the unique fixed point for the operator $F$ in the space of functions $\mathbb{B}$ i.e.
\be
F[\bar{\rho}]= \bar{\rho}
\ee
\end{theorem}

\begin{definition}
For $0\leq t\leq s\leq T$ and $x>0$, let $y=\log v(t,x)$, define:
\begin{eqnarray}
\bar{\alpha}(t,s,x) = \sE_t[U(\bar{c}_s \bar{X}_s)| \bar{X}_t=x] \quad ; \quad 
\bar{\delta}(t,s,y) = \sE_t[U_0(\bar{Y}_s)| \bar{Y}_t=y]
\end{eqnarray}
and
\be
G(t,x) := \int_t^T h(t,s) \bar{\alpha}(t,s,x)ds + h(t,T) \bar{\alpha}(t,T,x)
\ee
\end{definition}
The goal is to show that $G$ is a value function. We start by getting the PDEs for $\bar{\delta}, \bar{\alpha}$.

\begin{proposition}\label{PDE_alpha_delta}
$\bar{\delta}(t,s,y)$ satisfies the PDE:
\begin{equation}\label{PDE_bardelta}
\bar{\delta}_t+\frac{\theta^2}{2}\bar{\delta}_{yy}+(\bar{\rho}(t,y)-r-\frac{\theta^2}{2})\bar{\delta}_y(t,s,y) =0\;\; ; \;\; \bar{\delta}(s,s,y)=U_0(y)
\end{equation}
$\bar{\alpha}(t,s,x)$ satisfies the PDE:
\begin{equation}\label{PDE_baralpha}
\begin{cases}
\bar{\alpha}_t+\frac{\theta^2 v^2}{2v_x^2}\bar{\alpha}_{xx}+( r x -I(v(t,x))-\theta^2 \frac{v}{v_{x}}(t, x)   )\bar{\alpha}_x(t,s,x) =0\\\bar{\alpha}(s,s,x)=U_0(\log v(s,x))
\end{cases}
\end{equation}

\end{proposition}

For the proof, see the appendix.

In the next proposition, we show that $G_x=v$.
 
\begin{proposition}\label{Gx_equals_v}
For $t\in [0,T], x\in (0,\infty)$:
\be
G_x(t,x)=v(t,x)
\ee
\end{proposition}
The proof is in the appendix. The next result shows that $G(t,x)$ is a value function.

\begin{theorem}\label{subgame}
The function
$$G(t,x)=\sE_{t}\bigg[\int_{t}^{T} h(t,s)U(\bar{c}(s)\bar{X}(s))ds+ h(t,T)U(\bar{X}(T))\bigg]$$ is a value function.
It is the $C^{1,2}([0,T]\times (0,\infty))$ solution of the extended HJB
\begin{eqnarray}\label{extended_PDE_G}
&&G_t(t,x)+\sup_{(\pi, c) \textrm{ admissible}} \{  \frac{\sigma^2\pi^2 x^2}{2}G_{xx}(t,x)+(r-c+\theta\pi)xG_{x}(t,x)+U(x c)\} \nonumber \\
&&=\sE_{t}\bigg[\int_{t}^{T}\frac{\partial h(t,s)}{\partial t}U(\bar{c}(s)\bar{X}(s))ds+\frac{\partial h(t,T)}{\partial t}U(\bar{X}(T))\bigg] 
\end{eqnarray}

The subgame perfect wealth process is given by
$\bar{X}(s)=X^{\bar{\pi}, \bar{c}}(s)$.
\end{theorem}
The proof is in the appendix.

\section{ Algorithm for constructing the subgame perfect strategies}

We are given an investor with initial wealth $x_0$ at time $t=0$. We want to construct the subgame perfect strategies $\{\bar{c}_s, \bar{\pi}_s, \bar{X}_s, 0\leq s\leq T\}$ on the time interval $[0,T]$.

Recall $I=U'^{-1}$ is the inverse of the marginal utility, $I_0(y)=I(e^y)$ and $U_0(y)=U(I_0(y))$.

Step 1: Find the fixed point for the operator $F$.

For $\phi:[0,T]\times \sR \mapsto \sR $ be a bounded function with bounded continuous first derivatives in $t$ and $y$, define
 $Y_s^{t,y ;\phi}$ as the solution of the SDE
 \begin{equation*}
dY_{s}^{t, y; \phi}=(-\frac{\theta^2}{2}-r+\phi(s,Y_{s}^{\phi}))ds-\theta dW_{s} \;\; , \;\;  Y_t^{t,y ;\phi}=y 
\end{equation*}
Define the operator
 \begin{equation*}
F[\phi](t,y)= \frac{\frac{\partial }{\partial y} \sE_{t}\bigg[\int_{t}^{T}\frac{\partial h(t,s)}{\partial t}U_0(Y_s^{t,y; \phi})ds+\frac{\partial h(t,T))}{\partial t}U_0(Y_T^{t, y;\phi})\bigg]}{\frac{\partial }{\partial y} \sE_{t}\bigg[\int_{t}^{T}h(t,s)U_0(Y_s^{t,y;\phi})ds+h(t,T)U_0(Y_T^{t,y;\phi})\bigg]}
 \end{equation*}

Find a fixed point for the operator $F$ i.e. $\bar{\rho}$ such that:
\begin{equation*}
F[\bar{\rho}](t,y)= \bar{\rho}(t,y) \quad \forall t\in[0,T], y\in \sR
\end{equation*}
$\bar{\rho}$ can be found using successive approximations: $\phi_0 =0$ , $\phi_{n+1}=F[\phi_n]$ for all $n$.

We define the process $\bar{Y}_s=Y_s^{\bar{\rho}}$.

Step 2:  
Define the function $\bar{p}(t,y)$ by:

\begin{eqnarray*}
\bar{p}(t,y) &=& \sE\left[  e^{-r(T-t) } I_{0}(\bar{Y}_T+\theta^2(T-t))+\int_{t}^T e^{- r(s-t) } I_{0}(\bar{Y}_s+\theta^2(s-t))ds  \;\; \bigg|\; \bar{Y}_t=y \right]
\end{eqnarray*}

Given $t, y$, define 
$$c^{*}(t,y) = \frac{I_{0}(y)}{\bar{p}(t,y)} \quad , \quad \pi^{*}(t,y) = -\frac{\theta \bar{p}_y(t,y)}{\sigma \bar{p}(t,y)} $$

Step 3:
Given an initial wealth $x_0>0$ at time $t=0$, we need to find $y_0$ such that $x_0=\bar{p}(0, y_0)$.

For $s\in [0,T]$, with initial wealth $x_0$ at $t=0$, we can construct a subgame perfect strategy: 
$\{\pi^{*}(s,\bar{Y}_s^{0, y_0}), c^{*}(s,\bar{Y}_s^{0, y_0}), \bar{p}(s, \bar{Y}_s^{0, y_0}), 0\leq s\leq T \}$

%
%
%
%

\section{Appendix}

\subsection*{Preliminary inequalities and growth bounds}

  \begin{lemma}\label{estimates_I0_U0}
  Recall the relations $I_0(y):=I(e^y)$ and $U_0(y):=U(I_0(y))$. Since $U'\in \mathcal{D}_{0}(r_1,r_2) $, $I\in \mathcal{D}_{0}\left(\frac{1}{r_2},\frac{1}{r_1}\right) $
\be
\frac{1}{r_2} \leq -\frac{I_0'(y)}{I_0(y)} \leq \frac{1}{r_1}
\ee
 Furthermore, for $y, z\in \sR$:
$I_0$ satisfies:
 \begin{eqnarray}
 \frac{ |I_0(z)-I_0(y)|}{I_0(y)} \leq    \frac{1}{r_1} .  e^{\frac{ |z-y|}{r_1}} |z-y|\quad ; \quad \frac{ I_0(z)}{I_0(y)} &\leq&  e^{\frac{|z-y|}{r_1}} 
\end{eqnarray}
There is $\beta_1, \beta_2>0$ such that:
 \begin{eqnarray}  
\beta_1 e^{-\beta_2 |z-y|} \leq \frac{U_0'(z)}{U_0'(y)}\leq \beta_2 e^{\beta_2 |z-y|}
\end{eqnarray}

 \begin{eqnarray}
 \frac{ |U_0(z)-U_0(y)|}{|U_0'(y)|} &\leq&\beta_2 |z-y| .  e^{\beta_2 |z-y|}
\end{eqnarray}

 \end{lemma}
 
 \proof
 
 The first inequality comes from
  \begin{eqnarray*}
  \frac{ |I_0'(y)|}{I_0(y)} &=& |\frac{e^y }{I_0(y). U''(I_0(y))}|=\frac{1}{\mathcal{R}_1(I_0(y))}
 \end{eqnarray*}
$$  \frac{1}{r_2} \leq |\frac{ I_0'(y)}{I_0(y)} | \leq \frac{ 1}{r_1} $$
If we integrate between $y$ and $z$, we get:

$$e^{-\frac{|y-z|}{r_1}} \leq \frac{I_0(z)}{I_0(y)}\leq e^{\frac{|y-z|}{r_1}}$$

 By the mean value theorem, there is $z_0 \in [y, z]$ such that:
 \begin{eqnarray*}
  \frac{ |I_0(z)-I_0(y)|}{I_0(y)} &\leq& |z-y|  \frac{ |I_0'(z_0)|}{I_0(y)} \leq   \frac{1}{r_1} .  e^{\frac{ |z-y|}{r_1}} |z-y| \\
\end{eqnarray*}
 We also have
\be
r_1^2 e^{-(1+\frac{1}{r_1})|z-y|} \leq \frac{U_0'(z)}{U_0'(y)}= \frac{I_0'(z)e^z}{I_0'(y)e^{y}}\leq \frac{e^{(1+\frac{1}{r_1})|z-y|}}{r_1^2}
\ee
By the mean value theorem, there exists $z_0\in [y,z]$ such that:
 \begin{eqnarray*}
 \frac{ |U_0(z)-U_0(y)|}{|U_0'(y)|} &\leq& |z-y|  \frac{ |U_0'(z_0)|}{|U_0'(y)|} \leq  \frac{ |z-y|}{r_1^2} \cdot e^{(1+\frac{1}{r_1}) |z-y|}
\end{eqnarray*}

We take $\beta_1:= r_1^{2} \quad ; \quad \beta_2:= 1+\frac{1}{r_1^2}  $.

\ep

 \begin{lemma}
Suppose $Z $ is distributed as $ N(0,1)$ and $\kappa >0$. Then 
\be \label{bound_absexp}
\sE[e^{-\kappa |Z|}] = 2e^{\frac{\kappa^2}{2}} N(-\kappa) \quad ; \quad \sE[e^{\kappa |Z|}] = 2e^{\frac{\kappa^2}{2}} N(\kappa) 
\ee
where $N$ is the distribution function of the normal distribution.
Furthermore, for $x\geq 0$, according to \cite{Ch_Co_Mi1}, we have the lower bound:

\be \label{lower_N}
N(-x)\geq \frac{3}{4}\exp(-4 x^2)
\ee
\end{lemma}
\proof
We prove \eqref{bound_absexp} by direct computation of the expectation.

\ep

\subsection*{Fundamental solution and regularity}

The following lemma will allow us to get estimates for the operator $F$.

\begin{lemma}\label{quotients}
Let $n_1, n_2, q_1, q_2 \in \sR$ and suppose $q_1 q_2>0$.
The quotient $\frac{n_1+n_2}{q_1+q_2} $ is between $\frac{n_1}{q_1}$ and $\frac{n_2}{q_2}$ i.e.
$$\frac{n_1+n_2}{q_1+q_2}\in \bigg[\frac{n_1}{q_1}, \frac{n_2}{q_2}\bigg].$$
\end{lemma}

\proof
Let $x:=\frac{n_1+n_2}{q_1+q_2} $. We have: 
 \begin{eqnarray*}
x-\frac{n_1}{q_1} = \frac{n_2 q_1 - n_1 q_2}{q_1(q_1+q_2)} \quad & ; \quad x-\frac{n_2}{q_2} = \frac{n_1 q_2 - n_2 q_1}{q_2(q_1+q_2)} 
\end{eqnarray*}

$x-\frac{n_1}{q_1} $ and $x-\frac{n_2}{q_2}$ have opposite signs, therefore $x\in [ \frac{n_1}{q_1} , \frac{n_2}{q_2}]$.
 \ep

\begin{definition}
For $\phi: [0,T]\times \sR \mapsto \sR$ in $\mathbb{B}$ and $0\leq t_1\leq t_2\leq T$, we define the following quantities

\begin{eqnarray}
||\phi||_{t_1,t_2} &=& \sup_{(t,y)\in [t_1,t_2]\times \sR} |\phi(t,y)|  \\ 
||\phi_y||_{t_1,t_2} &=& \sup_{(t,y)\in [t_1,t_2]\times \sR} |\phi_y(t,y)|  
\end{eqnarray}

\end{definition}

For $\phi\in \mathbb{B}$, recall that $Y_u^{\phi; t,y}, t\leq u\leq T$ by the SDE:

\be
\begin{cases}
dY^{\phi; t,y}_{u}=(\phi(u,Y_{u}^{\phi; t,y}) -r -\frac{\theta^2}{2} )du-\theta dW_{u}, \\
 Y^{\phi; t,y}_t=y 
\end{cases}
\ee
Define $ \delta^{\phi}(t,s,y) :=  \sE_{t}^{\sP}\bigg[U_0({Y}^{\phi; t,y}_s)\bigg]$ and 
\begin{eqnarray}
F[\phi](t,y) &=&\frac{\int_{t}^{T}\frac{\partial h(t,s)}{\partial t}\delta_y^{\phi}(t,s,y) ds+\frac{\partial h(t,T)}{\partial t}\delta_y^{\phi}(t,T,y)}{ \int_{t}^{T}h(t,s)\delta_y^{\phi}(t,s,y) ds+h(t,T)\delta_y^{\phi}(t,T,y)}
\end{eqnarray}
Note that
\be
 \delta^{\phi}_y(t,s,y) :=  \sE_{t,y}^{\sP}\bigg[U_0'({Y}^{\phi}_s)e^{\int_t^s \phi_y(u,Y_u^{t,y})du}\bigg]<0
\ee
So the denominator in the expression of $F[\phi]$ is negative.

We want to find $\bar{\rho}$ such that $F[\bar{\rho}](t,y)=\bar{\rho}(t,y)$.
To prove it, we have to find a suitable space of functions $\phi(t,y)$ such that $F$ is a contraction and use a fixed point theorem.

The function $\delta^{\phi}$ satisfies the following PDE:
\be \label{PDE_delta}
\begin{cases}
\delta^{\phi}_t(t,s,y)+\frac{\theta^2}{2}\delta^{\phi}_{yy}(t,s,y)+(\phi(t,y)-r-\frac{\theta^2}{2})\delta^{\phi}_y(t,s,y)=0\\
\delta^{\phi}(s,s,y)=U_0(y)
\end{cases}
\ee
for all $t\in [0,s]$ and $y\in \sR$.

Write \be F= \frac{F_1}{F_0} \ee
where $F_1$ resp. $F_0$ are the numerator and denominator in the expression of $F$.

Define $$ F_y[\phi](t,y):= \frac{\partial F[\phi](t,y)}{\partial y}$$ 
We define $F_{0y}, F_{1y}$ in a similar manner.
We have:
$$F_y =   \frac{F_{1y}- F F_{0y} }{F_0} $$
i.e.
\be \frac{\partial F[\phi](t,y)}{\partial y} = \frac{\int_t^T (\frac{\partial h(t,s)}{\partial t} -h(t,s) F[\phi](t,y) ) \delta_{yy}^{\phi}(t,s,y) ds+ (\frac{\partial h(t,T)}{\partial t} -h(t,T) F[\phi](t,y) ) \delta_{yy}^{\phi}(t,T,y)}{\int_t^T h(t,s)\delta_y^{\phi}(t,s,y) ds+h(t,T)\delta_y^{\phi}(t,T,y) }
\ee

\begin{definition}
Let $\phi_1, \phi_2 \in \mathbb{B}_{\kappa}$.

For $0\leq t\leq s \leq T$ and $y\in \sR$:
\be
\epsilon^{\phi_1,\phi_2}(t,s,y) = \delta^{\phi_2}(t,s,y)-\delta^{\phi_1}(t,s,y)
\ee
\end{definition}

$\epsilon$ satisfies the PDE:
\begin{eqnarray}
&&\epsilon_t^{\phi_1,\phi_2}+\frac{\theta^2}{2}\epsilon_{yy}^{\phi_1,\phi_2}+(\phi_2(t,y)-r-\frac{\theta^2}{2}) \epsilon_y^{\phi_1,\phi_2}(t,s,y) =-(\phi_2(t,y)-\phi_{1}(t,y)) \delta_{y}^{\phi_1}(t,s,y)\\ \label{PDE_epsilon}
&&\epsilon^{\phi_1, \phi_2}(s,s,y) = 0
\end{eqnarray}
We want to get bounds for the quantities $|F[\phi_2](t,y)-F[\phi_1](t,y)|$,  $|F[\phi_2]_y(t,y)-F[\phi_1]_y(t,y)|$ in terms of 
$\epsilon^{\phi_1,\phi_2}$ and its derivatives.

\subsection*{Integral estimates}
\begin{proposition}\label{fundamental}
Let $b$ be a positive constant. Let $a(t,y)$ be a function defined on $[0,T]\times \sR$. Suppose $a$ is continuous, bounded and uniformly Lipschitz in the variable $y$ uniformly in $t$ i.e. $$|a(t,y_2)-a(t,y_1)|\leq ||a||_L |y_2-y_1|\quad \forall t\in [0,T],  y_1, y_2\in \sR.$$
Consider the Cauchy problem
\be \label{PDE_u}
-u_t(t,y)+b u_{yy}(t,y)+a(t,y) u_y(t,y) =0
\ee
Then for $l=0, 1, 2$, the fundamental solution $\Gamma(t, y, \tau, z)$, $0\leq \tau<t\leq T$, $y,z \in \sR$ satisfies the bound :
\be \label{estimate_fundamental}
|\frac{\partial^l}{\partial y^l} \Gamma(t,y; \tau, z) | \leq C (1+||a||_L )\times (t-\tau)^{-\frac{1+l}{2}} \exp(-c\frac{|y-\zeta|^2}{t-\tau})
\ee
where $c$, $C$ are constants independent of $t, y, ||a||_L$.
\end{proposition}
\proof
This result is a consequence of the construction of $\Gamma$ using the parametrix method as in \cite{Ga_Me1}.

Let $\Gamma_b(t,y)=\frac{1}{\sqrt{4\pi b t}} e^{-\frac{y^2}{4b t}}$ and 
$F(t,y,\tau,\zeta) = -a(t,y) \frac{\partial}{\partial y} \Gamma_b(y-\zeta, t-\tau)$. 

Let $Q(t,y,\tau, \zeta)$ satisfy the Volterra equation:
$$
Q(t,y,\tau, \zeta) = F(t,y,\tau,\zeta)+\int_{\tau}^t ds \int_{\sR} F(t,y,s,z) Q(s,z,\tau,\zeta) dz
$$
$$\Gamma_1(t,y,\tau,\zeta) = \int_{\tau}^t ds \int_{\sR} \Gamma_b(t-s, y-z) Q (s,z,\tau, \zeta) dz$$
Finally, the fundamental solution of the PDE \eqref{PDE_u} is
$$\Gamma(t,y, \tau, z)=\Gamma_b(t-\tau, y-z) +\Gamma_1(t,y,\tau,z)$$

The function $F$ satisfies a Lipschitz property in $y$, the function $Q$ satisfies a similar estimate.

The estimate \eqref{estimate_fundamental}  for $|\frac{\partial^l}{\partial y^l} \Gamma(t,y, \tau, z) |$ by following the proofs in Chapter 5 of \cite{Ga_Me1} [ Lemma 3.1, Lemma 3.3 and Theorem 3.5].

The constant $c$ depends on $b$ and $C$ depends on $b$ and $T$.
\ep

We will need the following calculations that will help us estimate the derivatives of parabolic PDE solutions.
\begin{lemma}\label{integral_exp}
For $ c >0$, $\tau>t\geq 0$.   $n_0 = 0, 1$ and $\alpha \in \sR$.
We  can get explicit expressions for the integral
\begin{equation}
I_{n_0, \alpha}^l = \int_{\sR}  (\tau-t)^{-\frac{1+l}{2}} |\zeta-y|^{n_0} e^{-c \frac{|y-\zeta|^2}{\tau-t} } e^{\alpha |\zeta-y|}d\zeta  
\end{equation}

\begin{equation}
I_{0, \alpha}^l =  (\tau-t)^{-\frac{l}{2}}\sqrt{ 2\pi} \frac{ e^{\frac{\alpha^2(\tau-t) }{4c}}}{\sqrt{2c}} N(\alpha \sqrt{\frac{\tau-t}{2c}}) 
\end{equation}

\begin{equation}
I_{1, \alpha}^l =\frac{  (\tau-t)^{\frac{1-l}{2}}}{c} \times \bigg(1+ \alpha \sqrt{ \pi} \sqrt{\frac{\tau-t}{c}}e^{\frac{\alpha^2(\tau-t) }{4c}} N(\alpha \sqrt{\frac{\tau-t}{2c}})  \bigg)
\end{equation}

\end{lemma}

\subsection*{Estimates for $\bar{\delta}$, $\epsilon$ and their derivatives.}
\begin{proposition}\label{estimates_deltaF}
Let $t\in [0,T], y\in \sR$, $\phi_1, \phi_2 \in \mathbb{B}$:
There exists $s_0, s_1, s_2, s_3\in [t,T]$ such that:
\begin{eqnarray}\label{delta_F}
 | F[\phi_2](t,y) - F[\phi_{1}](t,y)|& \leq & 2||\rho|| .  \frac{|\epsilon_y^{\phi_1, \phi_2}(t,s_0,y)|}{ |\delta_y^{\phi_1}(t,s_0,y)|} 
\end{eqnarray}

\begin{eqnarray}\label{delta_Fy}
| \frac{\partial F[\phi_2](t,y)}{\partial y}-\frac{\partial F[\phi_1](t,y)}{\partial y} | & \leq &  2||\rho|| . |  \frac{\epsilon_{yy}^{\phi_1,\phi_2}}{\delta_y^{\phi_1}}(t,s_1,y)|
+ 
4||\rho|| . |\frac{\delta_{yy}^{\phi_1}}{\delta_y^{\phi_1}}(t,s_2,y)| \times |\frac{\epsilon_y^{\phi_1,\phi_2}}{\delta_y^{\phi_1}}(t,s_3,y)|
\end{eqnarray}
\end{proposition}

\proof
\begin{eqnarray*}
&& F[\phi_2](t,y) - F[\phi_{1}](t,y)  =\frac{F_1[\phi_2]}{F_0[\phi_2]}(t,y)-\frac{F_1[\phi_1]}{F_0[\phi_1]}(t,y)\\
&&=\frac{F_1[\phi_2]-F_1[\phi_1](t,y)}{F_0[\phi_1](t,y)}-
 \frac{F[\phi_2](t,y)\times (F_0[\phi_2]-F_0[\phi_1](t,y))}{F_0[\phi_1](t,y)}\\
&& = \frac{ \int_t^T h(t,s) ( \rho_h(t,s)-F[\phi_{2}](t,y)) \epsilon_{y}^{\phi_1, \phi_2}(t,s,y)ds +h(t,T) ( \rho_h(t,T)-F[\phi_{2}](t,y)) \epsilon_y^{\phi_1, \phi_2}(t,T,y)}{\int_t^T h(t,s) \delta_y^{\phi_1}(t,s,y) ds+h(t,T) \delta_y^{\phi_1}(t,T,y)}
\end{eqnarray*}

By  lemma \ref{quotients}
\begin{eqnarray*}
&&  F[\phi_2](t,y) - F[\phi_{1}](t,y)  \in \bigg(  \frac{ \int_t^T h(t,s) ( \rho_h(t,s)-F[\phi_{2}](t,y)) \epsilon_{y}^{\phi_1, \phi_2}(t,s,y)ds }{\int_t^T h(t,s) \delta_y^{\phi_1}(t,s,y) ds},  (\rho_h(t,T)-F[\phi_{2}](t,y)) \frac{\epsilon_{y}^{\phi_1, \phi_2}(t,T,y) }{ \delta_y^{\phi_1}(t,T,y) } \bigg)
\end{eqnarray*}
and the extended mean value theorem yields the existence of $s_0 \in [t,T]$ such that:

\begin{eqnarray*}
&&| F[\phi_2](t,y) - F[\phi_{1}](t,y)| \leq |\frac{(\rho_h(t,s_0)-F[\phi_2](t,y)) . \epsilon_y^{\phi_1,\phi_2}(t,s_0,y)}{ \delta_y^{\phi_1}(t,s_0,y)}|
\end{eqnarray*}
Using the inequalities $|F[\phi]| \leq ||\rho|| $ and $|\rho_h | \leq ||\rho||$, we get inequality \eqref{estimates_deltaF}.

\begin{eqnarray*}
&& \frac{\partial F[\phi_2](t,y)}{\partial y}-\frac{\partial F[\phi_1](t,y)}{\partial y}  = \\
&&\frac{\int_t^T (\frac{\partial h(t,s)}{\partial t} -h(t,s) F[\phi_2](t,y) ) \delta_{yy}^{\phi_2}(t,s,y) ds}{\int_t^T h(t,s)\delta_y^{\phi_2}(t,s,y) ds}
- \frac{\int_t^T (\frac{\partial h(t,s)}{\partial t} -h(t,s) F[\phi_1](t,y) ) \delta_{yy}^{\phi_1}(t,s,y) ds}{\int_t^T h(t,s)\delta_y^{\phi_1}(t,s,y) ds}\\
&&= \frac{ \int_t^T h(t,s) \times \bigg[(F[\phi_1]-F[\phi_{2}](t,y))\delta_{yy}^{\phi_1}(t,s,y)+
(\rho_h(t,s)-F[\phi_{2}](t,y)) \epsilon_{yy}^{\phi_1,\phi_2}(t,s,y)\bigg]ds }{\int_t^T h(t,u) \delta_y^{\phi_1}(t,u,y) du}\\
&&+ \frac{\int_t^T h(t,s) (F[\phi_{1}](t,y)-\rho_h(t,s))\delta_{yy}^{\phi_1}(t,s,y)ds}{\int_t^T h(t,u) \delta_y^{\phi_1}(t,s,y) ds} \times \frac{ \int_t^s h(t,s)  \epsilon_{y}^{\phi_1,\phi_2}(t,s,y)ds}{\int_t^T h(t,s) \delta_y^{\phi_2}(t,s,y) ds}\bigg] 
\end{eqnarray*}
To simplify the calculations above, we have omitted the term in $(t,T,y)$.

We can apply the mean value theorem to the quotient of integrals and lemma \ref{quotients}:
  there exists $ s_0, s_1, s_2, s_3 \in [t,T]$ such that
\begin{eqnarray*}
|\frac{\partial F[\phi_2](t,y)}{\partial y}-\frac{\partial F[\phi_1](t,y)}{\partial y}|&& \leq |(\rho_h(t,s_1)-F[\phi_2](t,y))  \frac{\epsilon_{yy}^{\phi_1,\phi_2}}{\delta_y^{\phi_1}}(t,s_1,y)| \\
&&+ 
 |(-\rho_h(t,s_0)+F[\phi_2](t,y))  \frac{\epsilon_{y}^{\phi_1,\phi_2}}{\delta_y^{\phi_1}}(t,s_0,y) \; . \; \frac{\delta_{yy}^{\phi_1}}{\delta_y^{\phi_1}}(t,s_1,y)| \\
&&+
|(F[\phi_1](t,y)-\rho_h(t,s_2))\frac{\delta_{yy}^{\phi_1}}{\delta_y^{\phi_1}}(t,s_2,y)\times \frac{\epsilon_y^{\phi_1,\phi_2}}{\delta_y^{\phi_1}}(t,s_3,y)|
\end{eqnarray*}

Using the inequalities $|F[\phi]| \leq ||\rho|| $ and $|\rho_h | \leq ||\rho||$, we get the result.

\ep

We need to find upper bounds for the following quantities 
$|\frac{\delta_{yy}^{\phi_i}}{\delta_y^{\phi_i}}(t,s,y)|$, 
for $i=1, 2$ and for $|\frac{\epsilon_{y}^{\phi_1,\phi_2}}{\delta_y^{\phi_1}}(t,s,y)|$, $|\frac{\epsilon_{yy}^{\phi_1,\phi_2}}{\delta_y^{\phi_1}}(t,s,y)|$.

The following lemma establishes bounds for $\delta_y^{\phi}(t,s,y)$.

\begin{lemma}
Let $t, s \in [0,T], t\leq s, y\in \sR$,  $\phi \in \mathbb{B}_{\kappa}$:
there are positive constants $k_{10}, k_{20}$ independent of $\kappa$ such that
\be
 k_{10}e^{- \kappa(s-t)} \leq | \frac{ \delta_y^{\phi}(t,s,y)}{U_0'(y)}|  \leq  k_{20}e^{ \kappa(s-t)}
\ee

\be 
|\frac{ \delta_{yy}^{0}(t,s,y)}{U_0(y)}|   \leq  \beta_2 k_{20}  
\ee

\end{lemma}

\begin{proof}
Note that if $\phi \in \mathbb{B},$ 
\be
D(s) :=\frac{\partial Y_s^{\phi}}{\partial y} = \exp\left( \int_t^s \frac{ \partial \phi}{\partial y}(u,Y_u^{\phi}) du  \right)
\ee
Recall that $$ \delta^{\phi}(t,s,y) =  \sE_{t,y}^{\sP}\bigg[U_0({Y}^{\phi}_s)\bigg]$$ and
$$ \delta_y^{\phi}(t,s,y) =  \sE_{t,y}^{\sP}\bigg[U_0'({Y}^{\phi}_s)e^{\int_t^s \phi_y(u,Y_u^{\phi})du}\bigg]$$

thus,
\begin{eqnarray*}
| \frac{ \delta_y^{\phi}(t,s,y)}{U_0'(y)} | & \leq &  \sE_{t,y}^{\sP}\bigg[|\frac{U_0'({Y}^{\phi}_s)}{U_0'(y)}| e^{\int_t^s \phi_y(u,Y_u^{\phi})du}\bigg]
\leq  \beta_2\sE_t[ e^{ \beta_2 |Y_s^{\phi} -y | }  e^{\int_t^s \phi_{y} (u, Y_u^{\phi}) du} ] \\
  &\leq &2\beta_2 e^{ (\beta_2 (r+\frac{\theta^2}{2}+2||\rho||) + \frac{\beta_2^2 \theta^2}{2})(s-t)+ \kappa(s-t) }N(\beta_2 \theta\sqrt{s-t}) \\
 &\leq& 2 \beta_2 e^{ (\beta_2 (r+\frac{\theta^2}{2}+2||\rho||) + \frac{\beta_2^2 \theta^2}{2})(s-t)+\kappa (s-t)}
\end{eqnarray*}

Similarly,
\begin{eqnarray*}
| \frac{ \delta_y^{\phi}(t,s,y)}{U_0'(y)} | & \geq & 2\beta_1 e^{ -(\beta_2 (r+\frac{\theta^2}{2}+2||\rho||) + \frac{\beta_2^2 \theta^2}{2}-\kappa (s-t) }  N(-\beta_2 \theta\sqrt{s-t}) \\
 &\geq& \frac{3\beta_1}{2} e^{- \beta_2 (r+\frac{\theta^2}{2}+2||\rho||)(s-t) + \frac{\beta_2^2 \theta^2}{2})(s-t)-\kappa(s-t)-4\beta_2^2\theta^2(s-t)}
 \end{eqnarray*}
 where we used the inequality $N(-x)\geq \frac{3}{4}e^{-4x^2} \;\; \forall \; x>0$.
\end{proof}

Define the positive numbers \be
k_{10} :=  \frac{3\beta_1}{2}e^{ \big(-\beta_2 (r+\frac{\theta^2}{2}+2||\rho||) - \frac{7 \beta_2^2 \theta^2}{2}\big)T}\;\;;\;\;
k_{20} := 2\beta_2  e^{\big (\beta_2 (r+\frac{\theta^2}{2}+2||\rho||) + \frac{\beta_2^2 \theta^2}{2}\big)T}
\ee
We have the inequalities
\be
k_{10} e^{-\kappa (s-t) }\leq |\frac{ \delta_y^{\phi}(t,s,y)}{U_0'(y)} | \leq k_{20} e^{\kappa (s-t)}
\ee
Let $Y_s^0$ be the solution of the SDE \eqref{PDE_delta}, where we take $\phi=0$.

The functions $ \phi$ and $ \frac{\partial \phi}{\partial y} $ are bounded continuous functions of $(t,y)$. We can use Chapter 5, theorem 7.6 of (Brownian Motion and Stochastic Calculus):
let $R^{\phi} (t,y; \tau,\zeta) $  be the fundamental solution associated to $\delta^{\phi}$  given by the parabolic PDEs \eqref{PDE_delta}.
\begin{eqnarray*}
\delta^{\phi}(t,s,y)&=&\int_{\sR} R^{\phi}(y,s-t,\zeta,0) U_0(\zeta)  d\zeta 
\end{eqnarray*}

Then, by proposition \ref{fundamental}, we have the following estimate for $R^{\phi}$
\be
|\frac{ \partial^l R^{\phi} }{\partial y^l} (y,t,\zeta,\tau)| \leq C (1+\kappa) (\tau-t)^{-\frac{1+l}{2}} \exp\left(-c \frac{|y-\zeta|^2}{\tau-t}\right)
\ee
where $l=0, 1, 2$. The positive constants  $c, C$ depend on $r, \theta, T, || \phi||, \beta_2$ and do not depend on $\kappa$.

\begin{lemma}\label{bound_C0}
There is $k_{30}>0$ and $C_0>0$ depending only on $r, \theta, ||\rho||, T, \beta_2$ such that
$$|\frac{\delta_{yy}^0(t,s,y)}{U_0'(y)}|\leq k_{30} \quad; \quad |F[0]_y(t,y)|\leq C_0 \quad \forall t,y. $$
\end{lemma}

We can estimate $|\frac{\delta_{yy}^0(t,s,y)}{U_0'(y)}|$ easier because the fundamental solution will be of the form $R^{0} =R^0(s-t,y-\zeta)$.
Using an integration by parts, we get:
\begin{eqnarray*}
\delta_{yy}^0(t,s,y)&=& \int_{\sR} R_{yy}^0(y-\zeta,s-t) U_0(\zeta) d\zeta \\
&=& \int_{\sR} R_y^0(y-\zeta,s-t) U_0'(\zeta) d\zeta 
\end{eqnarray*}
Writing $$|R_y^0(y-\zeta,s-t)| \leq C (s-t)^{-\frac{1}{2}} \exp(-c \frac{|y-\zeta|^2}{s-t} )$$ 

$$
 |\frac{\delta_{yy}^0(t,s,y)}{U_0'(y)} |  \leq  \int_{\sR}
 \beta_2 e^{\beta_2 |\zeta-y|} C (s-t)^{-\frac{1}{2}} \exp(-c \frac{|y-\zeta|^2}{s-t} ) d\zeta = 
  C \beta_2  I_{0,\beta_2}^0 \leq C \beta_2 \sqrt{\frac{\pi}{c}} e^{\frac{\beta_2^2 T}{4c}}:= k_{30}
$$

\ep

If we study the expression of $F_y[\phi](t,y)$, we notice that the numerator and denominator are continuous functions of the variable $t,y$.
We can use the extended mean value theorem which is applied to a quotient of integrals.

By the mean value theorem there exists $s_0 \in [t,T]$ such that

\begin{eqnarray*}
|\frac{\partial F[\phi](t,y)}{\partial y}| \leq 2||\rho|| . |\frac{\delta_{yy}^{\phi}(t,s_0,y)}{\delta_y^{\phi}(t,s_0,y)}|
\end{eqnarray*}

For a fixed $t\in [0,T], s\in [t,T], y\in \sR$, we need to find bounds for $\frac{ \delta_{yy}^{\phi}(t,s,y)}{\delta_y^{\phi}(t,s,y) } $.

\begin{proposition}

Let $t, s \in [0,T], t\leq s, y\in \sR$,  $\phi \in \mathbb{B}_{\kappa}$:
There is a positive constant $C>0$ independent on $\kappa$ such that for all $0\leq t\leq s\leq T$:
\be
|\frac{ \delta_{yy}^{\phi}(t,s,y)}{U_0'(y)}|   \leq  k_{30}+C(1+\kappa)^2 \sqrt{s-t}
\ee
\end{proposition}
\proof
We write $\delta^{\phi}(t,s,y)=\delta^{0}(t,s,y)+\epsilon^{0,\phi}(t,s,y)$, and apply the triangular inequality:

\begin{eqnarray*}
|\frac{\delta_{yy}^{\phi}(t,s,y)}{U_0'(y)}|\leq |\frac{\delta_{yy}^{0}(t,s,y)}{U_0'(y)}|+|\frac{\epsilon_{yy}^{0,\phi}(t,s,y)}{U_0'(y)}|
\end{eqnarray*}
The first term is bounded by $k_{30}$, the second term can be rewritten

\begin{eqnarray*}
\epsilon_{yy}^{0,\phi}(t,s,y)&= &\int_t^s \int_{\sR} |R^{\phi}_{yy}(s-t,y,s-\tau,\zeta)
  \phi(\tau,\zeta) \delta^{0}_{y}(\tau,s,\zeta)| d\zeta d\tau \\
  &  = &\int_t^s \int_{\sR} R^{\phi}_{yy}(s-t,y,s-\tau,\zeta)(
   \phi(\tau,\zeta)\delta^{0}_{y}(\tau,s,\zeta) -   \phi(\tau,y) \delta^{0}_{y}(\tau,s,y)) d\zeta d\tau
\end{eqnarray*}
The term in parenthesis can be bounded using the mean value theorem. There is $\zeta_0\in [y,\zeta]$ such that:
 \begin{eqnarray*}
   |\phi(\tau,\zeta)\delta^{0}_{y}(\tau,s,\zeta) -   \phi(\tau,y) \delta^{0}_{y}(\tau,s,y))| & \leq & ( |\phi_y(\tau, \zeta_0)\delta^{0}_{y}(\tau,s,\zeta)| +|\phi(\tau,\zeta_0)\delta^{0}_{yy}(\tau,s,\zeta_0)|).|\zeta-y|\\
& \leq & (\kappa k_{20}  + ||\rho|| k_{30}) |U_0'(y)|. |\zeta-y| . e^{\beta_2 |\zeta-y|}
  \end{eqnarray*}

We use lemma \ref{integral_exp} to get:

\begin{eqnarray*}
|\frac{\epsilon_{yy}^{0,\phi}(t,s,y)}{U_0'(y)}|  &\leq & \int_t^s  \int_{\sR} C(1+\kappa)  (\tau-t)^{-3/2} e^{-{c} \frac{|y-\zeta|^2}{\tau-t}}  |\zeta-y| e^{\beta_2 |\zeta-y|}(\kappa k_{20} + ||\rho||  k_{30} )  d\zeta d\tau\\
& \leq & \frac{C}{c}\int_t^s  (\tau-t)^{-1/2} (1+\kappa) (\kappa k_{20} + ||\rho||  k_{30} ) \bigg(1+\beta_2 \sqrt{\frac{\pi (\tau-t)}{c}} e^{\frac{\beta_2^2(\tau-t)}{4 c}} \bigg) d\tau \\
|\frac{\epsilon_{yy}^{0,\phi}(t,s,y)}{U_0'(y)}|&\leq & \frac{2 C}{c}  ( k_{20} + ||\rho|| k_{30} )(1+\kappa)^2 \bigg(1+\beta_2 \sqrt{\frac{\pi (s-t)}{c}} e^{\frac{\beta_2^2(s-t)}{4 c}} \bigg) \times \sqrt{s-t}
\end{eqnarray*}
 
So there is $C>0$  depending on the coefficients of the PDE i.e. on $||\rho||, r, \theta, \beta_2, T$  such that 
\begin{eqnarray*}
|\frac{\delta_{yy}^{\phi}(t,s,y)}{U_0'(y)}|\leq  k_{30}+C(1+\kappa)^2 \sqrt{s-t}
\end{eqnarray*}
\ep

Now, let $\phi_1,\phi_2\in \mathbb{B}_{\kappa}$. 
We can estimate $|F[\phi_2]-F[\phi_1](t,y)|$ in terms of $|\epsilon_y^{\phi_1,\phi_2}|$ so we need an upper bound.

\begin{eqnarray*}
&& \epsilon^{\phi_1,\phi_2}(t,s,y) = - \int_t^{s} d\tau \int_{\sR}  R^{\phi_{2}}(y, s-t,\zeta,s-\tau)(\phi_{2}(\tau,\zeta)-\phi_{1}(\tau,\zeta)) \delta_y^{\phi_1}(\tau,s,\zeta) d\zeta
\end{eqnarray*}

\begin{eqnarray*}
&& \epsilon_y^{\phi_1,\phi_2}(t,s,y) = - \int_t^{s} d\tau \int_{\sR}  R_y^{\phi_{2}}(y,s-t,\zeta, s-\tau)(\phi_{2}(\tau,\zeta)-\phi_{1}(\tau,\zeta)) \delta_y^{\phi_1}(\tau,s,\zeta) d\zeta
\end{eqnarray*}

We write
$$ |(\phi_{2}(\tau,\zeta)-\phi_{1}(\tau,\zeta)) \delta_y^{\phi_1}(\tau,s,\zeta)| \leq |U_0'(\zeta)| . k_{20}e^{\kappa (s-\tau)}   |\phi_{2}(\tau,\zeta)-\phi_{1}(\tau,\zeta)|
$$
and 

$$ |R_y^{\phi_{2}}(y,s-t,\zeta,s-\tau)|\leq  C(1+\kappa) (\tau-t)^{-1} e^{-c \frac{|y-\zeta|^2}{\tau-t}} $$ to get
\begin{eqnarray*}
 |\frac{\epsilon_y^{\phi_1,\phi_2}(t,s,y)}{U_0'(y)}| &\leq & \int_t^{s} \int_{\sR}  C (1+\kappa) (\tau-t)^{-1} e^{-c \frac{|y-\zeta|^2}{\tau-t}} . e^{\beta |\zeta-y|} .  k_{20}e^{\kappa (s-\tau)} ||\phi_{2}-\phi_{1}||_{t,s}  \;\; d\zeta d\tau \\
& \leq & \int_t^s C(1+\kappa) (\tau-t)^{-\frac{1}{2}}\sqrt{ \pi} \frac{ e^{\frac{\beta^2(\tau-t) }{4c_1}}}{\sqrt{c}} N(\beta \sqrt{\frac{\tau-t}{2c}})  .k_{20}e^{\kappa(s-t)} ||\phi_{2}-\phi_{1}||_{t,s} d\tau \\
|\frac{\epsilon_y^{\phi_1,\phi_2}(t,s,y)}{U_0'(y)}|&\leq & C .(1+\kappa)e^{\kappa(s-t)}  \sqrt{s-t} .  ||\phi_{2}-\phi_{1}||_{t,s}
\end{eqnarray*}

for a constant $C>0$ depending only on $r, \theta, T, \beta_2, ||\rho||$.

Thus
$$  |F[\phi_2]-F[\phi_1](t,y)| \leq 2||\rho|| \frac{C e^{2\kappa (T-t)}  . (1+\kappa) }{k_{10}} . \sqrt{T-t} \times  ||\phi_{2}-\phi_{1}||_{t,T}
$$

We have
\begin{eqnarray*}
&& \epsilon_{yy}^{\phi_1,\phi_2}(t,s,y) = - \int_t^{s} d\tau \int_{\sR}  R_{yy}^{\phi_{2}}(y,s-t,\zeta, s-\tau)(\phi_{2}(\tau,\zeta)-\phi_{1}(\tau,\zeta)) \delta_y^{\phi_1}(\tau,s,\zeta) d\zeta \\
&&=  - \int_t^{s} d\tau \int_{\sR}  R_{yy}^{\phi_{2}}(y,s-t,\zeta, s-\tau)\bigg((\phi_{2}(\tau,\zeta)-\phi_{1}(\tau,\zeta)) \delta_y^{\phi_1}(\tau,s,\zeta) -(\phi_{2}(\tau,y)-\phi_{1}(\tau,y)) \delta_y^{\phi_1}(\tau,s,y)   \bigg) d\zeta 
\end{eqnarray*}

By the mean value theorem, there is $\zeta\in [y,\zeta]$ such that the term in parenthesis is equal to
$$
(\zeta-y)\big((\phi_{2y}-\phi_{1y})\delta_y^{\phi_1}(\tau,s,\zeta_0)+(\phi_{2}-\phi_{1})\delta_{yy}^{\phi_1}(\tau,s,\zeta_0)\big)
$$
and it can be bounded by
$$
|\zeta-y| . |U_0'(\zeta_0)| . \big(||\phi_{2y}-\phi_{1y}||_{t,s}k_{20}e^{\kappa(s-t)}+||\phi_{2}-\phi_{1}||_{t,s} (k_{30}+C(1+\kappa)^2 \sqrt{s-t})\big)
$$
We now use the inequalities
\begin{eqnarray*}
&& |R_{yy}^{\phi_{2}}(y,s-t,\zeta,s-\tau)|\leq  C (1+\kappa) (\tau-t)^{-3/2} e^{-{c} \frac{|y-\zeta|^2}{\tau-t}} 
 \end{eqnarray*}
  and Lemma \ref{integral_exp}
 to get

\begin{eqnarray*}
|\frac{ \epsilon_{yy}^{\phi_1,\phi_2}(t,s,y)}{U_0'(y)}| &\leq &  \int_t^{s} d\tau \int_{\sR}  C(1+\kappa) (\tau-t)^{-3/2} e^{-{c} \frac{|y-\zeta|^2}{\tau-t}}  |\zeta-y| . e^{\beta |\zeta-y|} \\
&& \times \big(||\phi_{2y}-\phi_{1y}||_{t,s}k_{20}e^{\kappa(s-t)}+||\phi_{2}-\phi_{1}||_{t,s} (k_{30}+C(1+\kappa)^2 \sqrt{s-t})\big)
 d\zeta \\
 & \leq & \frac{2 C (s-t)^{\frac{1}{2}}}{c} \times \bigg(1+ \beta \sqrt{\frac{\pi(s-t)}{c}}e^{\frac{\beta^2(s-t) }{4c}}   \bigg)\times (1+\kappa)\\
 && \times \big(||\phi_{2y}-\phi_{1y}||_{t,s}k_{20}e^{\kappa(s-t)}+||\phi_{2}-\phi_{1}||_{t,s} (k_{30}+C(1+\kappa)^2\sqrt{s-t})\big)\\
|\frac{ \epsilon_{yy}^{\phi_1,\phi_2}(t,s,y)}{U_0'(y)}| & \leq  &C' (1+\kappa)^3 e^{\kappa(s-t)} \sqrt{s-t} \times \big(||\phi_{2y}-\phi_{1y}||_{t,s}+||\phi_{2}-\phi_{1}||_{t,s} \big)
\end{eqnarray*}
for a positive constant $C'$  depending on the coefficients of the PDE i.e. on $||\rho||, r, \theta, \beta_2, T$.

\subsection*{Estimates for the operator $F$}
We combine the estimates for $|\frac{ \epsilon_{y}^{\phi_1,\phi_2}(t,s,y)}{U_0'(y)}|$, $|\frac{ \epsilon_{yy}^{\phi_1,\phi_2}(t,s,y)}{U_0'(y)}|$ with proposition \ref{estimates_deltaF} to get the following proposition.

\begin{proposition}
Let $\kappa>0$, $\phi_1, \phi_2 \in \mathbb{B}_{\kappa}$.
There are universal constants $K_0, K_1>0$ independent of $\kappa$, $t, y$ such that for all $t\in [0,T]$, $ y\in \sR$:

\be |F[\phi_2]-F[\phi_1](t,y)| \leq K_0 (1+\kappa)e^{2\kappa(T-t)} \sqrt{T-t} \times  ||\phi_2- \phi_1||_{t,T}
\ee
\be
|F[\phi_2]_y-F[\phi_1]_y(t,y)| \leq  K_1(1+\kappa)^3 e^{2\kappa(T-t)} \sqrt{T-t} \times \big(||\phi_2-\phi_1||_{t,T} + ||\phi_{2y}-\phi_{1y}||_{t,T} \big)
\ee
\end{proposition}

Recall that by Lemma \ref{bound_C0}, there is $C_0>0$ independent of $\kappa, t, y$ such that $|F[0]_y(t,y) | \leq C_0$ $\forall (t,y)\in [0,T]\times \sR$.
If we choose 
\be
\kappa = \max(||\rho||, 2 C_0)
 \ee
and $\nu \in (0,T)$ small enough such that:

\begin{eqnarray}
K_0 (1+\kappa)e^{2\kappa T} \sqrt{\nu} \leq \frac{1}{4}\quad ; \quad  K_1(1+\kappa)^3 e^{2\kappa T} \sqrt{\nu} \leq  \frac{1}{4}
\end{eqnarray}
We define the following space of functions which is a restriction of $\mathbb{B}_{\kappa}$ to $[T-\nu, T]\times \sR$.

\be  \label{BK_1}
\mathbb{B}_{\kappa}^{(1)} = \{ \phi:[T-\nu,T]\times \sR \rightarrow \sR,  \textrm{ continuous in $t,y$;  $C^1$ in $y$; $|\phi(t,y)|\leq ||\rho||$ ; $|\phi_y(t,y)|\leq \kappa$} \}
\ee

For $\phi, \phi_1, \phi_2 \in \mathbb{B}_{\kappa}^{(1)}$, we have:
\begin{eqnarray*}
&&||F[\phi_2]_y-F[\phi_1]_y||_{T-\nu, T} + ||F[\phi_2]-F[\phi_1]||_{T-\nu, T}   \leq \frac{1}{2}( ||\phi_{2y}-\phi_{1y}||_{T-\nu, T} +  ||\phi_{2}-\phi_{1}||_{T-\nu, T} )
 \end{eqnarray*}

Furthermore, 
$F[\phi](t,y)$ and $F[\phi]_y(t,y)$ are quotients of integrals and therefore are continuous functions of $t, y$.

We also have:
$$||F[\phi ]_y||_{T-\nu, T} \leq ||F[\phi ]_y-F[0]_y||_{T-\nu, T} +||F[0]_y||_{T-\nu, T} \leq \frac{\kappa}{4}+\frac{1}{4}
\big(||\rho|| + \kappa \big)\leq \kappa
$$

Thus, $F$ is a contraction for the norm $\phi \mapsto || \phi ||_{T-\nu, T}+||\phi_y||_{T-\nu,T}$ in the complete convex subset $\mathbb{B}_{\kappa}^{(1)}$,  therefore it admits a unique fixed point $\Phi$.

\be
F[\Phi](t,y) =  \Phi(t,y) \quad \forall (t,y)\in [T-\nu, T]\times \sR 
\ee

We can build $\Phi$ on $[T-2\nu, T]\times \sR$ by defining the space
$\mathbb{B}_{\kappa}^{(2)}$ of functions $ \phi:[T-2\nu,T]\times \sR \rightarrow \sR$ such that:
\begin{itemize}
\item $\phi(t,y)$ , $\phi_y(t,y)$ are continuous in $t,y$.
\item $|\phi(t,y)|\leq ||\rho||$ ; $|\phi_y(t,y)|\leq \kappa$ ; $\phi(t,y) = \Phi(t,y)\quad \forall t\in [T-\nu,T], y\in \sR$.
\end{itemize}

The function $$f(t,y) = \begin{cases}
 \Phi(t,y) \textrm{ for } (t,y) \in  [T-\nu, T]\times \sR  \\
  \Phi(T-\nu,y)  \textrm{ for } (t,y) \in [T-2\nu, T-\nu)\times \sR
\end{cases}
$$
We can construct $\Phi$ in this way on intervals $[T-(k+1)\nu, T-k\nu]$ for $k=1,2, \cdots$.

$\Phi \in \mathbb{B}_{\kappa}$ and $F[\Phi](t,y) = \Phi(t,y)$ for $t\in [0,T], y\in \sR$.

We take $\bar{\rho} := \Phi$.  
And this ends the proof of the fixed point result.

\proof{Proof proposition \ref{p_v_D0}}

Note that $I_0'(y)=e^y I'(e^y)$.

Since $U' \in \mathcal{D}_0(r_1, r_2)$, this implies that its inverse $I\in  \mathcal{D}_0\left(\frac{1}{r_2}, \frac{1}{r_1}\right)$ i.e. $\frac{1}{r_2} \leq -\frac{xI'(x)}{I(x)}\leq \frac{1}{r_1}$. 
Thus,
\begin{equation}\label{bound_I0}
 \frac{1}{r_2} \leq -\frac{I_0'(y)}{I_0(y)}\leq \frac{1}{r_1} \quad \forall y\in \sR
\end{equation}
We also have:
\begin{eqnarray*}
-\frac{\bar{p}_y(t,y)}{\bar{p}(t,y)} &&= - \frac{ \sE_t\left[  \int_{t}^T e^{\int_t^s ( \bar{\rho}_y(u, \bar{Y}_u)-r ) du} I_0'(\bar{Y}_s+\theta^2(s-t))ds  + e^{\int_t^T  (\bar{\rho}_y(u, \bar{Y}_u)-r ) du} I_0'(\bar{Y}_T+\theta^2(T-t))\right]}{ \sE_t\left[ \int_{t}^T e^{- r(s-t) }  I_0(\bar{Y}_s+\theta^2(s-t))ds  +e^{- r(T-t) }  I_0(\bar{Y}_T+\theta^2(T-t))  \right]}
\end{eqnarray*}
And since for all $t,y$,  $|\bar{\rho}_y(t,y) | \leq \kappa$, we get the estimate:
\be
 \frac{1}{r_2} e^{-\kappa(T-t)}\leq -\frac{\bar{p}_y(t,y)}{\bar{p}(t,y)} \leq  \frac{1}{r_1}   e^{\kappa(T-t)}
\ee

And since $\bar{p}(t,y) = p(t,e^y)$, we get
\be                                                                                                                                                                                                                                                                                                                                                                                           
\frac{1}{r_2(t)} \leq -\frac{xp_x(t,x)}{p(t,x)} \leq \frac{1}{r_1(t)}
\ee
i.e. $p(t, .) \in \mathcal{D}_0\left(\frac{1}{r_2(t)},     \frac{1}{r_1(t)} \right)$.

This implies $v(t, .)$ the $x$-inverse of $p(t,.)$ is well defined.
By Lemma \ref{lemma_D0}, $v(t, .) \in \mathcal{D}_0(r_1(t), r_2(t))$.
\ep

\proof{Proposition \ref{bounds_barp}}
We want to show that $|\frac{1}{I_0(y)}\times \frac{\partial^l  \bar{p}(t,y)}{\partial y^l} |$ is bounded. 
The proof is similar to the proof of the boundedness of $|\frac{1}{U_0'(y)}\times \frac{\partial^l  \bar{\delta}}{\partial y^l} |$. We just have to replace the boundary function $U_0$ by $I_0$ and 
the constant $\beta_2$ by $\frac{1}{r_1}$ as the estimates of Lemma \ref{estimates_I0_U0} show.
\ep

\proof{Proposition \ref{bounds_candpi}}

Let $y_1\leq y_2$, by integrating the expression \eqref{bound_I0} between $y_1$ and $y_2$: 

\begin{eqnarray*}
\frac{1}{r_2} . (y_2-y_1) \leq  -\log (I_0(y_2)/I_0(y_1))  \leq \frac{1}{r_1} . (y_2-y_1)\\
 e^{- \frac{y_2-y_1}{r_1}}\leq  \frac{I_0(y_2)}{I_0(y_1)}  \leq  e^{- \frac{y_2-y_1}{r_2}}
\end{eqnarray*}
Next, we get a bound for $\frac{\bar{p}(t,y)}{I_0(y)}$.
\begin{eqnarray*}
\frac{\bar{p}(t,y)}{I_0(y)} &&=   \sE_t\left[  \int_{t}^T e^{-r (s-t)} \frac{I_0(\bar{Y}_s+\theta^2(s-t))}{I_0(y)} ds  + e^{-r(T-t)} \frac{I_0(\bar{Y}_T+\theta^2(T-t))}{I_0(y)}\right]
\end{eqnarray*}
We write 
\begin{eqnarray*}
\sE_t e^{\frac{-1}{r_1}. |\bar{Y}_s+\theta^2(s-t))-y|} \leq \sE_t\big[\frac{I_0(\bar{Y}_s+\theta^2(s-t))}{I_0(y)}\big] \leq \sE_t e^{\frac{1}{r_1}. |\bar{Y}_s+\theta^2(s-t))-y|}
\end{eqnarray*}

\begin{eqnarray*}
2e^{\frac{-(r+\frac{\theta^2}{2}+||\rho||)}{r_1}.+\frac{\theta^2}{2r_1^2})(s-t)}N( - \theta \sqrt{s-t}/r_1   ) \leq \sE_t\big[\frac{I_0(\bar{Y}_s+\theta^2(s-t))}{I_0(y)}\big] \leq 2e^{\frac{(r+\frac{\theta^2}{2}+||\rho||}{r_1}.+\frac{\theta^2}{2r_1^2})(s-t)}N(  \theta \sqrt{s-t}/r_1   )\\
\frac{1}{2}e^{\big(-\frac{r+\frac{\theta^2}{2}+||\rho||}{r_1}-\frac{3\theta^2}{2r_1^2}\big)(s-t)} \leq \sE_t\big[\frac{I_0(\bar{Y}_s+\theta^2(s-t))}{I_0(y)}\big] \leq 2e^{\big(\frac{r+\frac{\theta^2}{2}+||\rho||}{r_1}+\frac{\theta^2}{2r_1^2}\big)(s-t)}
\end{eqnarray*}

Taking the integral between $t$ and $T$ yields $r_3(t)$ and $r_4(t)$ such that
$$0< r_3(t) \leq \frac{\bar{p}(t,y)}{I_0(y)} \leq r_4(t)$$
\ep

\proof{Proposition \ref{PDE_alpha_delta}}

We use the fact that  $|U_0'(y)| \leq Ce^{c|y|}$ for constants $c, C>0$, By integrating, we get that
 $|U_0(y)|$ is bounded by a function of the form $Ce^{c|y|}$.

The PDE \eqref{PDE_bardelta} has a unique $C^{1,2}$ solution therefore, we can apply the Feynman Kac formula to get the result.

For $\bar{\alpha}$, we notice that $a_2(t,x):=\frac{\theta^2 v^2}{2x^2 v_x^2}$ is bounded away from zero and $a_1(t,x):=r  - \frac{I(v(t,x))}{x}-\theta^2 \frac{v}{x v_{x}}(t, x)$ is bounded.

We get a PDE of the form: 
$$ f_t(t,x)+ a_2(t,x)x^2 f_{xx}(t,x)+a_1(t,x) x f_x(t,x)=0$$ 

We integrate the inequality $r_1(t)\leq -\frac{xv_x(t,x)}{v(t,x)}\leq r_2(t)$ between 1 and $x$ to get

\begin{eqnarray*}\label{boundsv}
v(t,1) x^{-r_2(t)}\leq v(t,x)\leq v(t,1) x^{-r_1(t)} \quad \text{if }\;\; x\geq 1\\
v(t,1) x^{-r_1(t)}\leq v(t,x)\leq v(t,1) x^{-r_2(t)} \quad \text{if }  \;\; 0< x< 1\\
\end{eqnarray*}

The boundary $|U_0(\log v(s,x))|$ can be bounded by $C(1+x^{c_1}+x^{-c_2})$ for constants $C, c_1, c_2>0$ so the PDE \eqref{PDE_baralpha} 
has a unique $C^{1,2}$ solution. By the Feynman-Kac theorem, we get that 
$\bar{\alpha}$ satisfies the PDE \eqref{PDE_baralpha}.

\ep

\begin{proposition}\label{smooth_solution}

Let $f$, $g$ be two functions defined respectively on $[0,T]\times \sR$ and $\sR$. We suppose $f$ and $g$ are continuous and satisfy a growth condition $|f(t,y)| +|g(y)| \leq C e^{c|y|} $ for all $t\in [0,T], y\in \sR$.
Let $u$ be a classical solution of the heat equation 
\be
-u_{t}(t,y) + u_{yy}(t,y) + f(t,y)=0 ; u(0,y) =g(y)
\ee
Then $u$ is smooth on $(0,T]\times \sR$.
\end{proposition}


The proof can be found in \cite{Ev1}, Section 2.3, Theorem 8.
Next, we prove that $G_x=v$.

\proof{Proposition \ref{Gx_equals_v}

Define \be
q(t,y) := e^y \bar{p}_y(t,y) \quad ; \quad  L(t,y) =G(t,\bar{p}(t,y))
\ee
Since $y\mapsto \bar{p}(t,y)$ is a bijection from $\sR$ onto $(0,\infty)$, it is enough to show that $L_y(t,y)=q(t,y)$ for all $t, y$.

We can see that by doing an affine change of variables, we can reduce the PDEs defining $\bar{p}$ and $\bar{\delta}$ to the heat equation with a source term. We can then use proposition \ref{smooth_solution} to establish the smoothness of $\bar{p}$ and $\bar{\delta}$. 

$\bar{p}_{ty}, \bar{p}_{yyy}$, $\bar{\delta}_{ty}, \bar{\delta}_{yyy}$ are well defined. 
We also have:
$$
q_t = e^y \bar{p}_{ty} \;;\;\; q_y = e^y \bar{p}_{y}+e^y  \bar{p}_{yy} = q+e^y  \bar{p}_{yy} \quad ; \quad  q_{yy}=-q+ 2q_y+e^y \bar{p}_{yyy}
$$
$q$ satisfies the PDE:

\be
q_t +\frac{\theta^2}{2}q_{yy}+(\bar{\rho}(t,y)-r-\frac{\theta^2}{2})q_y-\rho(t,y) q(t,y)+e^yI_0'(y) =0 \; ; \; q(T,y)=e^yI_0'(y)
\ee
Since 
$$\bar{\alpha}(t,s,\bar{p}(t,y)) = \bar{\delta}(t,s,y)$$ we get

$$L(t,y) = \int_t^T h(t,s) \bar{\delta}(t,s,y)ds + h(t,T) \bar{\delta}(t,T,y)$$
and $L_y$ satisfies the PDE
\begin{eqnarray*}
&&L_{ty}+\frac{\theta^2}{2}L_{yyy}+(\bar{\rho}(t,y)-r-\frac{\theta^2}{2})L_{yy}-\bar{\rho}(t,y) L_y +e^yI_0'(y)\\
&& = -\bar{\delta}_y(t,t,y) + h(t,T) [\bar{\delta}_t +\frac{\theta^2}{2}\bar{\delta}_{yy}+(\bar{\rho}(t,y)-r-\frac{\theta^2}{2})\bar{\delta}_y(t,T,y)]\\
 && +\int_t^T \frac{\partial h}{\partial t}(t,s) \bar{\delta}_y(t,s,y)ds + \frac{\partial h}{\partial t} (t,T) \bar{\delta}_y(t,T,y)+\int_t^T h(t,s) \times [ \bar{\delta}_t +\frac{\theta^2}{2}\bar{\delta}_{yy}+(\bar{\rho}(t,y)-r-\frac{\theta^2}{2})\bar{\delta}_y(t,s,y)] ds  \\
&&= -e^y I_0'(y) +e^yI_0'(y)  -\bar{\rho}(t,y) L_y (t,y) + \int_t^T \frac{\partial h}{\partial t}(t,s) \bar{\delta}_y(t,s,y)ds + \frac{\partial h}{\partial t} (t,T) \bar{\delta}_y(t,T,y)\\
&&= (F[\bar{\rho}](t,y)-\bar{\rho}(t,y) )\times L_y (t,y)\\
&& =0
\end{eqnarray*}
And since $\bar{\rho}(t,y) = F[\bar{\rho}](t,y)$, the last quantity is equal to zero.
 
We also have $L_y(T,y) = \bar{\delta}_y(T,T,y)=U_0'(y) = e^y I_0'(y)$.

The functions $L_y$ and $q$ are solutions of the same parabolic PDE and they are bounded by $Ce^{c|y|}$ for positive constants $c, C$ that are big enough, therefore by uniqueness of the solution of the PDE:
$$
L_y(t,y)=q(t,y)
$$
Since $L_y(t,y) = \bar{p}_y(t,y) G_x(t, \bar{p}(t,y)) $, $q(t,y)=e^y \bar{p}_y(t,y)$ and $\bar{p}(t,y)=p(t,e^y)$, we get:

$
G_x(t, \bar{p}(t,y)) = e^y 
$
i.e.
$G_x(t,p(t,x)) = x$ , $ \forall x>0$.

And since $p(t,x)$ is the $x$-inverse of $v(t,x)$, we conclude that $G_x(t,x) = v(t,x)$ for all $x>0$.

\ep

 \proof{Theorem \ref{barXbarp}}
 
 The SDE
 \be \label{SDE_barX}
 d\bar{X}_s=
 [r-\bar{c}(s,\bar{X}_s)+\theta \sigma\bar{\pi}(s,\bar{X}_s)]\bar{X}_s ds+\sigma \bar{\pi}(s,\bar{X}_s)\bar{X}_s dW_s
 \ee
  has continuously differentiable coefficients with bounded derivatives.
 To see that, we express the derivatives $\frac{\partial (x\bar{c}(t,x))}{\partial x}, \frac{\partial (x\bar{\pi}(t,x))}{\partial x}$ in terms of $x=\bar{p}(t,y)$.
 
 \begin{eqnarray*}
 \frac{\partial (x\bar{c}(t,x))}{\partial x} = \frac{\partial (I(v(t,x))}{\partial x}=v_x(t,x) . I'(v(t,x)) = \frac{v_x(t,x)}{U''(I(v(t,x))}
 \end{eqnarray*}
 Since $v(t,\bar{p}(t,y)) =e^y$, taking the $y$-derivative yields:
 \be
 \bar{p}_y(t,y) v_x(t,\bar{p}(t,y)) = e^y
 \ee
 \begin{eqnarray*}
 \frac{\partial (x\bar{c}(t,x))}{\partial x} = \frac{e^y}{\bar{p}_y(t,y) U''(I_0(y))} =  \frac{U'(I_0(y)) I_0(y)}{\bar{p}_y(t,y) I_0(y) U''(I_0(y))}
 \end{eqnarray*}
 Then, we use the fact that both $\frac{U'(I_0(y)) }{I_0(y) U''(I_0(y))}$ and $\frac{I_0(y)}{\bar{p}_y(t,y) }$ are bounded independently of $t, y$ to conclude that there is $L_1>0$ such that \begin{eqnarray*}
| \frac{\partial (x\bar{c}(t,x))}{\partial x} | \leq L_1 \quad \forall t,x
 \end{eqnarray*}
 Similarly, in terms of the variable $y$,   $$\frac{\partial }{\partial y}= \frac{\partial x}{\partial y}. \frac{\partial }{\partial x} = \bar{p}_y(t,y) . \frac{\partial }{\partial x}$$
 
 \begin{eqnarray*}
 \frac{\partial (x\bar{\pi}(t,x))}{\partial x} = -\frac{\theta}{\sigma} \frac{\partial }{\partial x}(\frac{v(t,x)}{v_x(t,x)} = -\frac{\theta}{\sigma \bar{p}_y(t,y)} \frac{\partial }{\partial y}\big(\frac{e^y\bar{p}_y(t,y)}{e^y}\big)= -\frac{\theta \bar{p}_{yy}(t,y)}{\sigma \bar{p}_y(t,y)} 
  \end{eqnarray*}
 By proposition \ref{bounds_barp}, $-\frac{\theta \bar{p}_{yy}(t,y)}{\sigma \bar{p}_y(t,y)} $ is uniformly bounded.
There is $L_2>0$ independent of $t, x$ such that: 
$$
| \frac{\partial (x\bar{\pi}(t,x))}{\partial x} | \leq L_2 \quad \forall t,x
$$
   So, by Chapter V, Theorem 39 of \cite{Pr1},
the SDE \eqref{SDE_barX} has a unique classical solution.

 For $t\leq s\leq T$, we have:
 \be
 \bar{X}_s = x \exp\bigg(\int_t^s \bigg( r-\bar{c}+\theta \sigma\bar{\pi}-\frac{\sigma^2 (\bar{\pi})^2(u,\bar{X}_u)}{2} \bigg)du+\int_t^s \sigma \bar{\pi}(u,\bar{X}_u) dW_u\bigg)
 \ee
  By Proposition \ref{bounds_candpi}, $\bar{\pi}(t,x), \bar{c}(t,x)$ are uniformly bounded independently of $x$.
So 

$\forall s\in [t,T]$ : $\bar{X}_s<\infty$ $\sP$ a.s.

  Note that:
  \be
  \bar{\pi}(t,\bar{p}(t,y))=-\frac{\theta}{\sigma} \frac{\bar{p}_y(t,y)}{\bar{p}(t,y)}
  \ee
  
 By Ito's lemma: 
 \begin{eqnarray*}
 d\bar{p}(s,\bar{Y}_s) &=& [\bar{p}_t+\frac{\theta^2}{2}\bar{p}_{yy}+(\bar{\rho}(s,\bar{Y}_{s})-\frac{\theta^2}{2}-r)\bar{p}_y ] ds -\theta \bar{p}_y(s,\bar{Y}_{s})dW_s\\
 & = &
 (r\bar{p}(s,\bar{Y}_{s})-I_0(\bar{Y}_s)-\theta^2 \bar{p}_y(s,\bar{Y}_{s}))ds-\theta \bar{p}_y(s,\bar{Y}_{s})dW_s
 \end{eqnarray*}
 In the last equality, we have used PDE \eqref{PDE_barp}.
Thus,
  $$d\bar{p}(s,\bar{Y}(s))=
 [r-\bar{c}(s,\bar{p}(s,\bar{Y}_s))+\theta \sigma\bar{\pi}(s,\bar{p}(s,\bar{Y}_s))]\bar{p}(s,\bar{Y}_s) ds+\sigma \bar{\pi}(s,\bar{p}(s,\bar{Y}_s))\bar{p}(s,\bar{Y}_s) dW_s$$

 If we choose $y=\bar{Y}_t$ such that  $\bar{p}(t,y)=x=\bar{X}_t$,  the two processes  $\bar{X}_s$ and $\bar{p}(s,\bar{Y}_s)$ satisfy the same SDE. By
  uniqueness of the solution, we have the equality almost surely:
 $$ \bar{X}_s =\bar{p}(s,\bar{Y}_s) \quad \forall s\in [t,T]. $$
 
 The relation $y=\log v(t,x)$ comes from the fact that $p$ is the $x$ inverse of $v$.
 \ep

\proof{ Theorem \ref{subgame} }

We calculate the partial derivatives of $G(t,x)=\int_t^T h(t,s) \bar{\alpha}(t,s,x) ds+ h(t,T)\bar{\alpha}(t,T,x)$. 
\begin{eqnarray}\label{Gt}
G_t(t,x) &=& -h(t,t) \bar{\alpha}(t,t,x) +\int_t^T h(t,s) \bar{\alpha}_t(t,s,x) ds+ h(t,T)\bar{\alpha}_t(t,T,x) \\
 &+& \int_t^T \frac{\partial h(t,s)}{\partial t} \bar{\alpha}(t,s,x) ds+ \frac{\partial h(t,T)}{\partial t} \bar{\alpha}(t,T,x)
\end{eqnarray}
\begin{eqnarray}\label{Gx}
G_x(t,x) &=& \int_t^T h(t,s) \bar{\alpha}_x(t,s,x) ds+ h(t,T)\bar{\alpha}_x(t,T,x) 
 \end{eqnarray}
\begin{eqnarray}\label{Gxx}
G_{xx}(t,x) &=& \int_t^T h(t,s) \bar{\alpha}_{xx}(t,s,x) ds+ h(t,T)\bar{\alpha}_{xx}(t,T,x) 
 \end{eqnarray}
 Recall that $\bar{\alpha}$ satisfies equation \eqref{PDE_baralpha}:
$$
\bar{\alpha}_t+\frac{\theta^2 v^2}{2v_x^2}\bar{\alpha}_{xx}+( r x -I(v(t,x))-\theta^2 \frac{v}{v_{x}}(t, x)   )\bar{\alpha}_x(t,s,x) =0\;\; ; \;\; \bar{\alpha}(s,s,x)=U(I(v(s,x)))
$$
therefore
\begin{eqnarray*}
&&G_t(t,x)+\frac{\theta^2 v^2}{2v_x^2}G_{xx}(t,x)+\big( r x -I(v(t,x))-\theta^2 \frac{v}{v_{x}}(t, x)   \big)G_x(t,x)+U(I(v(t,x)) \\
&&=  \int_t^T \frac{\partial h(t,s)}{\partial t} \bar{\alpha}(t,s,x) ds+ \frac{\partial h(t,T)}{\partial t} \bar{\alpha}(t,T,x)
\end{eqnarray*}
Reorganizing, and using the fact that $G_x=v$, we get equation \eqref{extended_PDE_G}.

We conclude by noting that $G$ is strictly increasing ($G_x =v>0$) and strictly concave ( $G_{xx}(t,x)=v_x(t,x)=\frac{1}{p_x(t,v(t,x))} <0$) so the sup in the extended HJB is attained at $\bar{\pi}, \bar{c}$.  
$G$ is a solution of the extended HJB.
\begin{eqnarray*}
&&G_t(t,x)+\sup_{\pi,c} \{ \Ac^{{\pi},{c}}G(t,x)+U(x c)\} \\
&&=\sE_{t}\bigg[\int_{t}^{T}\frac{\partial h(t,s)}{\partial t}U(\bar{c}(s)\bar{X}(s))ds+\frac{\partial h(t,T)}{\partial t}U(\bar{X}(T))\bigg]
\end{eqnarray*}

We can then use the verification theorem for the extended HJB system.
The function
$G(t,x)$, along with the controls $\bar{\pi}, \bar{c}$ is a solution of the extended HJB system to conclude that 

$\{\bar{\pi}(s,\bar{X}_s), \bar{c}(s,\bar{X}_s), \bar{X}_s, 0\leq s\leq T\}$ defines a subgame perfect strategy.

We have thus constructed the value function $V(t,x)=G(t,x)$.
\ep

\label{End}

\label{NumDocumentPages}

\end{document}